\documentclass{llncs}%article}
\usepackage{wrapfig}

\usepackage[sort&compress,numbers]{natbib}%colon %%%% natbib package must be before algorithm2e
\usepackage[boxed,ruled,vlined,linesnumbered]{algorithm2e}
\SetKwHangingKw{Local}{Local}
\SetKwHangingKw{Global}{Global}
\SetKwHangingKw{Constant}{Constant}
\SetKwHangingKw{Input}{Input: }
\SetKwHangingKw{Variable}{Variable}
\SetKwHangingKw{Output}{Output: }
\SetKwHangingKw{DataStructure}{Data Structure: }
\SetKwHangingKw{Oper}{Operation: }
\SetKwHangingKw{Field}{Field: }
\SetKwHangingKw{Inter}{Interface: }
\SetKwHangingKw{Proced}{Procedure: }
\SetKwHangingKw{Funct}{Function: }
\SetKwHangingKw{ExternalEvent}{External event}
\SetKwHangingKw{Receive}{Upon Receive }
\SetKwFor{For}{for}{do}{endfor}
\SetKwHangingKw{From}{from }
\SetKwHangingKw{Action}{Action :}
\SetKwHangingKw{Constant}{Constant :}

\usepackage{url}
\usepackage[bookmarks=false,colorlinks=true,citecolor={black},filecolor={black},linkcolor={black},urlcolor={black},pdfstartview={XYZ null null 1.22}]{hyperref}

\usepackage{times}
\usepackage{theorem}
\usepackage{amssymb}
\usepackage{xspace}
\newcommand{\qedClaim}{\hfill \ensuremath{\Box}}
\newenvironment{proofClaim}{\noindent {\bf Proof.}\ }{\qedClaim\par\vskip 4mm\par}

\usepackage{subfigure}% subfig} %
\newcommand\newcaptionstyle[2]{%
  \expandafter\ifx\csname caption@@#1\endcsname\relax
    \defcaptionstyle{#1}{#2}%
  \else
    \PackageError{caption}{Caption style `#1' already defined}{}%
  \fi}
\newcommand\defcaptionstyle[2]{%
  \@namedef{caption@@#1}{#2}}

\newcommand{\B}{\vspace*{-\smallskipamount}}
\newcommand{\BB}{\vspace*{-\medskipamount}}
\newcommand{\BBB}{\vspace*{-\bigskipamount}}

\newcommand{\FFF}{\vspace*{\bigskipamount}}

\newcommand{\BigO}[0]{{\cal O}\xspace}

\newcommand{\op}[1]{{\textsf{#1}}}
\newcommand{\remove}[1]{}

\newcommand{\Section}[1]{\section{#1}}

\newcommand{\Subsection}[1]{\noindent {\bf #1}~~~}
\begin{document}
%\initfloatingfigs
%\mainmatter
%\linenumbers

\renewcommand{\thefootnote}{\fnsymbol{footnote}}

\title{Self-Stabilizing Byzantine Resilient\\ Topology Discovery and Message Delivery}
\author{Shlomi Dolev~\thanks{Partially supported by
Deutsche Telekom, Rita Altura Trust Chair in Computer Sciences,
Lynne and William Frankel Center for Computer Sciences, Israel
Science Foundation (grant number 428/11) and Cabarnit Cyber Security MAGNET Consortium.}
\and Omri Liba~$^\ast$ \and Elad M.\ Schiller~\thanks{Partially supported by the EC, through project FP7-STREP-288195, KARYON (Kernel-based ARchitecture for safetY-critical cONtrol) and the European Union Seventh Framework Programme (FP7/2007-2013) under grant agreement No. 257007.}}

\institute{Department of Computer Science,
Ben-Gurion University of the Negev,
Beer-Sheva, Israel.  \email{ \{dolev, liba\}@cs.bgu.ac.il } \and Department of Computer Science and Engineering, Chalmers University of Technology, Goeteborg, Sweden. \email{elad@chalmers.se}}

\date{}
\maketitle

\begin{abstract}
Traditional Byzantine resilient algorithms use $2f + 1$ vertex-disjoint paths to ensure message delivery in the presence of up to $f$ Byzantine nodes. The question of how these paths are identified is related to the fundamental problem of topology discovery.
Distributed algorithms for topology discovery cope with a never ending task: dealing with frequent changes in the network topology and unpredictable transient faults. Therefore, algorithms for topology discovery should be self-stabilizing to ensure convergence of the topology information following any such unpredictable sequence of events. We present the first such algorithm that can cope with Byzantine nodes. Starting in an arbitrary global state, and in the presence of $f$ Byzantine nodes, each node is eventually aware of all the other non-Byzantine nodes and their connecting communication links.
Using the topology information, nodes can, for example, route messages across the network and deliver messages from one end user to another. We present the first deterministic, cryptographic-assumptions-free, self-stabilizing, Byzantine-resilient algorithms for network topology discovery and end-to-end message delivery. We also consider the task of $r$-neighborhood discovery for the case in which $r$ and the degree of nodes are bounded by constants. The use of $r$-neighborhood discovery facilitates polynomial time, communication and space solutions for the above tasks.
The obtained algorithms can be used to authenticate parties, in particular during the establishment of private secrets, thus forming public key schemes that are resistant to man-in-the-middle attacks of the compromised Byzantine nodes. A polynomial and efficient end-to-end algorithm that is based on the established private secrets can be employed in between periodical secret re-establishments.
\end{abstract}

\renewcommand{\thefootnote}{\arabic{footnote}}

\Section{Introduction}
Self-stabilizing Byzantine resilient topology discovery is a fundamental
distributed task that enables communication among parties in the network
even if some of the components are compromised
by an adversary. Currently, such topology discovery is becoming extremely
important  where countries'
main infrastructures, such as the electrical smart-grid, water
supply networks and intelligent transportation systems
are subject to cyber-attacks.
Self-stabilizing Byzantine resilient algorithms naturally
cope with mobile attacks \citep[e.g.,][]{DBLP:conf/podc/OstrovskyY91}.
Whenever the set of compromised components is fixed (or dynamic, but small)
during a period that suffices for convergence of the algorithm, the system
starts demonstrating useful behavior following the convergence.
For example, consider the case in which nodes of the smart-grid are
constantly compromised by an adversary while local recovery
techniques, such as local node reset and/or refresh, ensure
the recovery of a compromised node after a bounded time.
Once the current compromised set does not imply a partition of
the communication graph, the distributed control of the smart
grid automatically recovers.
Self-stabilizing Byzantine resilient algorithms for topology discovery
and message delivery are important for systems that have to cope
with unanticipated transient violations of the
assumptions that the algorithms are based upon,
such as unanticipated violation of the upper number of compromised nodes
and unanticipated transmission interferences that is beyond the error
correction code capabilities.

The dynamic and difficult-to-predict nature of electrical
smart-grid and intelligent transportation systems give rise to
many fault-tolerance issues and require efficient solutions.
Such networks are subject to transient faults due to
hardware/software temporal malfunctions or short-lived
violations of the assumed settings for the location and
state of their nodes. Fault-tolerant systems that
are {\em self-stabilizing}~\cite{D2K}
can recover after the occurrence of transient faults, which can
drive the system to an arbitrary system state. The system
designers consider {\em all} configurations as possible
configurations from which the system is started.
The self-stabilization design criteria liberate the system designer
from dealing with specific fault scenarios, risking neglecting
some scenarios, and having to address each fault scenario separately.

We also consider Byzantine faults that address the possibility of
a node to be compromised by an adversary and/or to run a corrupted
program, rather than merely assuming that they start in an arbitrary local state.
Byzantine components may behave arbitrarily (selfishly, or even maliciously) as message senders
and as relaying nodes. E.g., Byzantine nodes may block messages, selectively omit messages,
redirect message routes, playback messages, or modify messages. Any system behavior is possible, when all (or one third or more of) the nodes are Byzantine nodes. Thus, the number of Byzantine nodes, $f$, is usually restricted to be less than one third of the nodes~\cite{L96,D2K}.

The task of $r$-{\em neighborhood network discovery} allows each node
to know the set of nodes that are at most $r$ hops away from it in
the communication network. Moreover, the task provides information
about the communication links attached to these nodes.
The task {\em topology discovery} considers knowledge regarding
the node's entire connected component.
The $r$-neighborhood network discovery and network topology discovery tasks
are identical when $r$ is the communication graph radius.

This work presents the first deterministic self-stabilizing algorithms
for $r$-neighborhood discovery in the presence of Byzantine nodes.
We assume that every $r$-neighborhood cannot be partitioned by the
Byzantine nodes. In particular, we assume the existence of at least $2f+1$ vertex-disjoint
paths in the $r$-neighborhood, between any two non-Byzantine nodes,
where at most $f$ Byzantine nodes are present in the $r$-neighborhood, rather than in the entire network.~\footnote{Section~\ref{s:e2e} considers cases in which $r$ and an upper bound on the node degree, $\Delta$, are constants. For these cases, we have $\BigO(n)$ disjoint $r$-neighborhoods. Each of these (disjoint) $r$-neighborhoods may have up to $f$ Byzantine nodes, and yet the above assumptions about at least $2f+1$ vertex-disjoint paths in the $r$-neighborhood, hold.}
Note that by the self-stabilizing nature of our algorithms, recovery is guaranteed
after a temporal violation of the above assumption.
When $r$ is defined to be the communication graph radius, our assumptions are
equivalent to the standard assumption for Byzantine agreement in general (rather than only complete)
communication graphs. In particular the standard assumption is that $2f+1$ vertex
disjoint paths exist and {\it are known} (see e.g., \cite{L96}) while we present
distributed algorithms to find these paths starting in an arbitrary state.

\Subsection{Related work.}
Self-stabilizing algorithms for finding vertex-disjoint paths
for at most two paths between any pair of nodes,
and for all vertex-disjoint paths in anonymous mesh networks appear in~\cite{DBLP:conf/sss/Al-AzemiK11}
and in~\cite{DBLP:journals/comcom/HadidK09}, respectively. We propose self-stabilizing Byzantine resilient procedures for finding $f+1$ vertex-disjoint paths in $2f+1$-connected graphs.
In~\cite{DBLP:conf/wdag/DuboisMT11}, the authors study the problem of spanning tree
construction in the presence of Byzantine nodes.
Nesterenko and Tixeuil~\cite{DBLP:journals/tpds/NesterenkoT09} presented a
{\em non-stabilizing} and inconsistent algorithm for topology discovery in the presence of Byzantine nodes -- see the paper's errata for further details about the algorithm's flaws.~\footnote{~\url{http://vega.cs.kent.edu/~mikhail/Research/topology.errata.html}}
%In addition the algorithm flaws~\footnote{~\url{http://vega.cs.kent.edu/~mikhail/Research/topology.errata.html}}, the authors do not consider the automatic recovery implied by the self-stabilization property.
%
Awerbuch and Sipser~\cite{DBLP:conf/focs/AwerbuchS88} consider algorithms that were designed for synchronous static network and give topology update as an example. They show how to use such algorithms in asynchronous dynamic networks. Unfortunately, their scheme starts from a consistent state and cannot cope with transient faults or Byzantine nodes.

The problems of {\em Byzantine gossip}~\citep{DBLP:journals/dc/MinskyS03,DBLP:conf/osdi/LiCWNRAD06,DBLP:conf/wdag/DolevGGN07,DBLP:journals/sigops/AlvisiDGKLRT07,DBLP:journals/adhoc/BurmesterLY07,DBLP:conf/spaa/FernandessM08}
and {\em Byzantine Broadcast}~\citep{DBLP:conf/dsn/DrabkinFS05,DBLP:journals/ijfcs/PaquetteP06}
consider the dissemination of information in the presence of Byzantine nodes rather than self-stabilizing topology discovery. Non-self-stabilizing Byzantine resilient gossip in the presence of one selfish node is considered in~\citep{DBLP:conf/osdi/LiCWNRAD06,DBLP:journals/sigops/AlvisiDGKLRT07}. In~\citep{DBLP:conf/wdag/DolevGGN07} the authors study oblivious deterministic gossip algorithms for multi-channel radio networks with a malicious adversary. They assume that the adversary can disrupt one channel per round, preventing communication on that channel. In~\citep{DBLP:journals/adhoc/BurmesterLY07} the authors consider probabilistic gossip mechanisms for reducing the redundant transmissions of flooding algorithms. They present several protocols that exploit local connectivity to adaptively correct propagation failures and protect against Byzantine attacks.
Probabilistic gossip mechanisms in the context of recommendations and social networks are considered in~\citep{DBLP:conf/spaa/FernandessM08}. In~\citep{DBLP:journals/dc/MinskyS03} the authors consider rules
for avoiding a
combinatorial explosion in (non-self-stabilizing) gossip protocol.
Note that deterministic and self-stabilizing solutions are not presented
in~\citep{DBLP:journals/dc/MinskyS03,DBLP:conf/osdi/LiCWNRAD06,DBLP:conf/wdag/DolevGGN07,DBLP:journals/sigops/AlvisiDGKLRT07,DBLP:journals/adhoc/BurmesterLY07,DBLP:conf/spaa/FernandessM08}.
%
%Brahms is a probabilistic non-self-stabilizing algorithm for uniform node sampling inlarge dynamic systems prone to Byzantine behavior~\citep{DBLP:journals/cn/BortnikovGKKS09}.
Drabkin et al.~\citep{DBLP:conf/dsn/DrabkinFS05} consider non-self-stabilizing broadcast
protocols that overcome Byzantine failures by using digital signatures, message
signature gossiping, and failure detectors. Our deterministic self-stabilizing algorithm merely
use the topological properties of the communication graph to ensure correct message delivery to the application layer in the presence of message omission, modifications and Byzantine nodes. A non-self-stabilizing broadcasting algorithm is considered in~\citep{DBLP:journals/ijfcs/PaquetteP06}. The authors assume the restricted case in which links and nodes of a communication network are subject to Byzantine failures, and that faults are distributed randomly and independently.

%Update propagation in time logarithmic in the number of replicas and linear in the number of corrupt replicas is considered in~\citep{DBLP:conf/wdag/MalkhiPS01}.

\Subsection{Our contribution.}
We present two cryptographic-assumptions-free yet secure algorithms that are deterministic,
self-stabilizing and Byzantine resilient.

We start by showing the existence of deterministic, self-stabilizing, Byzantine resilient algorithms
for network topology discovery and end-to-end message delivery.
The algorithms convergence time is in $\BigO(n)$.
They take in to account every possible path and requiring bounded (yet exponential)
memory and bounded (yet exponential) communication costs.
Therefore, we also consider the task of $r$-neighborhood discovery, where $r$ is a constant.
We assume that if the $r$-neighborhood of a node has $f$ Byzantine nodes, there are $2f+1$
vertex independent paths between the node and any non-Byzantine node in its $r$-neighborhood.
The obtained $r$-neighborhood discovery algorithm requires polynomial memory and communication costs and supports
deterministic, self-stabilizing, Byzantine-resilient algorithm for end-to-end message delivery
across the network. Unlike topology update, the proposed end-to-end message delivery algorithm establishes message exchange synchronization between end-users that is based on message reception acknowledgments.

%\Subsection{Document structure.}
%%
%Settings and requirements appear in Section~\ref{s:pre}.
%The self-stabilizing Byzantine resilient distributed algorithm for topology discovery is presented in
%Section~\ref{s:SSBRTD}. The end-to-end communication algorithm appears in Section~\ref{s:e2e}.
%Extensions and concluding remarks appear in Section~\ref{s:ext}.
Detailed proofs appear in the Appendix.

%\pagebreak

\Section{Preliminaries}
\label{s:pre}

We consider settings of a standard asynchronous system \citep[cf.][]{D2K}.
The system consists of a set, $N = \{ p_i \}$, of communicating entities, chosen from a set, $P$,
which we call {\em nodes}.
The upper bound on the number of nodes in the system is $n=|P|$.
Each node has a unique identifier. Sometime we refer to a set, $P \setminus N$, of nonexisting nodes that a false indication on their existence can be recorded in the system. A node $p_i$ can directly communicate with its {\em neighbors}, $N_i \subseteq N$. The system can be represented by an undirected network of directly communicating nodes, $G=(N,E)$, named the {\em communication graph}, where $E= \{ (p_i, p_j) \in N \times N : p_j \in N_i \}$. We denote $N_k$'s set of indices by $indices(N_k) = \{ m : p_m \in N_k\}$ and $N_k$'s set of edges by $edges(N_j) = \{ p_j \} \times N_j$.

The $r$-neighborhood of a node $p_i \in N$ is the connected component that includes $p_i$ and all nodes that can be reached from $p_i$ by a path of length $r$ or less.
The problem of $r$-neighborhood topology discovery
considers communication graphs in which $p_i$'s degree, $\delta_i$, is bounded by a constant $\Delta$.
Hence, when both the neighborhood radius, $r$, and the node degree, $\Delta$, are constants the number
of nodes in the $r$-neighborhood is also bounded by a constant, namely by $\BigO(\Delta^{r+1})$.

We model the communication channel, $queue_{i,j}$, from node $p_i$ to
node $p_j \in N_i$ as a FIFO queuing list of the messages
that $p_i$ has sent to $p_j$ and $p_j$ is about to receive. When $p_i$
sends message $m$, the operation \op{send} inserts a copy of
$m$ to the queue $queue_{i,j}$ of the one destination $p_j$, such that $p_j \in N_i$. We assume that the number of messages
in transit, i.e., stored in $queue_{i,j}$, is at most $capacity$.
Once $m$ arrives, $p_j$ executes \op{receive} and $m$ is dequeued.

We assume that $p_i$ is completely aware of $N_i$, as in~\citep{DBLP:journals/tpds/NesterenkoT09}.
In particular, we assume that the identity of the sending node is known to the receiving one.
In the context of the studied problem, we say that node $p_i \in N$ is {\em correct} if it reports
on its genuine neighborhood, $N_i$. A {\em Byzantine} node, $p_{b} \in N$, is a node
that can send arbitrarily corrupted messages.
Byzantine nodes can introduce new messages and modify or omit messages that pass through them.
This way they can, e.g., disinform correct nodes about their neighborhoods, about the neighborhood of other correct nodes, or the path through which messages travel,
to name a very few specific misleading actions that Byzantine nodes may exhibit.
Note that our assumptions do not restrict system settings in which a {\em duplicitous Byzantine} node, $p_b$, reports about $N_b$ differently to its correct neighbors. In particular, $p_b$ can have $\{N_{b_1}, \ldots N_{b_{\delta_b}}\}$ reports, such that $p_b$'s identity in $N_{b_i}$ is different than the one in $N_{b_j}$, where $\delta_x$ is the degree of node $p_x$. One may use a set of non-duplicitous Byzantine nodes, $\{p_{b_1}, \ldots p_{b_\delta}\}$, to model each of $p_b$'s reports. Thus, for a $2k+1$ connected graph, the system tolerates no more than $\lfloor k/\Delta \rfloor$ duplicitous Byzantine nodes, where $\Delta$ is an upper bound on the node degree.

%Note that, in a system setting where this assumption does not hold, there is a need for a tighter bound on the number of Byzantine nodes. Namely, if the graph is $2k+1$ connected than we can tolerate no more than $\lfloor k/\delta \rfloor$ Byzantine nodes where $\delta$ denotes the maximum neighborhood size.
%

We denote $C$ and $B$ to be, respectively, the set of correct and Byzantine nodes.
We assume that $|B|= f$, the identity of $B$'s nodes is unknown to the ones in $C$, and $B$ is fixed throughout the considered execution segment. These execution segments are long enough for convergence and then for obtaining sufficient useful work. We assume that between any pair of correct nodes there are at least $2f+1$ vertex-disjoints paths.
We denote by $G_c=(C, E \cap C \times C)$ the {\em correct graph} induced by the set of correct nodes.

%%%%%%%%%%%%%%%%%%%%%%%%%%%%%%%%%%%%%%%%%%%%%%%%%%%%%%%%%%%%

Self-stabilizing algorithms never terminate~\cite{D2K}.
The non-termination property can be easily identified in the code of a self-stabilizing algorithm:
the code is usually a do forever loop that contains communication operations with the neighbors.
An iteration is said to be complete if it starts in the loop's first line and ends at
the last (regardless of whether it enters branches).
%A node may executes other parts of the program, or other programs, and periodically activates the loop.

Every node, $p_i$, executes a program that is a sequence of {\em
(atomic) steps}. For ease of description, we assume the interleaving
model with atomic step execution; a single step at any
given time. An input event can either be the receipt of a
message or a periodic timer going off triggering $p_i$ to \op{send}.
Note that the system is totally asynchronous and
the (non-fixed) node processing rates are irrelevant
to the correctness proof.

%A timer going off can be a trigger for a local
%a spontaneous action, such as a spontaneous \op{send} operation.

The {\em state} $s_i$ of a node $p_i$ consists of the value of
all the variables of the node (including the set of all incoming
communication channels, $\{ queue_{j,i} | p_j \in N_i \}$. The
execution of a step in the algorithm can change the state of a
node.
The term {\em (system) configuration} is used for a tuple of the form
$(s_1,s_2,\cdots,s_n)$, where each $s_i$ is the state of node
$p_i$ (including messages in transit for $p_i$).
We define an {\em execution} $E={c[0],a[0],c[1],a[1],\ldots}$ as an
alternating sequence of system configurations $c[x]$ and steps $a[x]$,
such that each configuration $c[x+1]$ (except the initial
configuration $c[0]$) is obtained from the preceding configuration
$c[x]$ by the execution of the step $a[x]$.
We often associate the notation of a step with its executing
node $p_i$ using a subscript, e.g., $a_i$.
An execution $R$ (run) is {\em fair} if every correct node, $p_i \in C$,
executes a step infinitely often in $R$.
Time (e.g. needed for convergence) is measured by the number of
{\em asynchronous rounds}, where the first asynchronous round is
the minimal prefix of the execution in which every node takes at
least one step. The second asynchronous round is the first asynchronous round
in the suffix of the run that follows the first asynchronous round, and so on.
The message complexity (e.g. needed for convergence) is the number of messages
measured in the specific case of synchronous execution.

We define the system's task by a set of executions called \emph{legal
executions} ($LE$) in which the task's requirements hold. A
configuration $c$ is a \emph{safe configuration} for an algorithm and
the task of $LE$ provided that any execution that starts in $c$ is a legal
execution (belongs to $LE$). An algorithm is \emph{self-stabilizing}
with relation to the task $LE$ when every infinite execution of the
algorithm reaches a safe configuration with relation to the algorithm
and the task.

\Section{Topology Discovery} %Self-Stabilizing Byzantine Resilient
\label{s:SSBRTD}
The algorithm learns about the neighborhoods that the nodes report. Each report message contains an ordered list of nodes it passed so far, starting in a source node. These lists are used for verifying that the reports are sent over $f+1$ vertex-disjoint paths.

%Each node $p_i$ periodically sends a message to each neighbor. The message sent contains the local topology, a source $i$ and an empty path.
When a report message, $m$, arrives to $p_i$, it inserts $m$ to the queue $informedTopology_i$, and tests the queue consistency until there is enough independent evidence to support the report. %The result array is initialized just prior to the consistency test.
The consistency test of $p_i$ iterates over each node $p_k$ such that, $p_k$ appears in at least one of the messages stored in $informedTopology_i$.
For each such node $p_k$, node $p_i$ checks whether there are at least $f+1$ messages from the same source node that have mutually vertex-disjoint paths and report on the same neighborhood. The neighborhood of each such $p_k$, that has at least $f+1$ vertex-disjoint paths with identical neighborhood, is stored in the array $Result_i[k]$ and the total number of paths that relayed this neighborhood is kept in $Count[k]$.

We note that there may still be nodes $p_{fake} \in P \setminus (N)$, for which there is an entry $Result[fake]$.
For example, $informedTopology$ may contain $f$ messages, all originated from different Byzantine nodes, and a message $m^{\prime}$ that appears in the initial configuration and supports the (false) neighborhood
the Byzantine messages refer to. These $f+1$ messages can contain mutually vertex-disjoint paths, and thus during the consistency test, a result will be found for $Result[fake]$. We show that during the next computations, the message $m^{\prime}$ will be identified and ignored.
The $Result$ array should include two reports for each (undirected) edge; the two nodes that are attached to the edge, each send a report. Hence, $Result$ includes a set of directed (report) edges. The term {\em contradicting edge} is needed when examining the $Result$ set consistency.

\BB\begin{definition}[Contradicting edges]
\label{d:ContradictingEdge}
Given two nodes, $p_i, p_j \in P$, we say that the edge $(p_i, p_j)$ is {\em contradicting with the set $evidence \subseteq edges(N_j)$}, if $(p_i, p_j) \not \in evidence$.
\BB\B \end{definition}

Following the consistency test, $p_i$ examines the $Result$ array for contradictions. Node $p_i$ checks the path of each message $m \in informedTopology_i$ with source $p_r$, neighborhood $neighborhood_r$ and $Path_r$. If every edge $(p_s,p_j)$ on the path appears in $Result[s]$ and $Result[j]$,  then we move to the next message. Otherwise,
we found a fake supporter, and therefore we reduce $Count[r]$ by one. If the resulting $Count[r]$ is smaller than $f+1$, we nullify the $r$'th entry of the $Result$ array.
Once all messages are processed, the $Result$ array consisting of the (confirmed) local topologies is the output.
At the end, $p_i$ forwards the arriving message, $m$, to each neighbor that does not appear in the path of $m$. The message sent by $p_i$ includes the node from which $m$ arrived as part of the path $m$.

\Subsection{The pseudocode of Algorithm~\ref{algo:discovery}}
%
%The pseudocode appears in Algorithm~\ref{algo:discovery}.
In every iteration of the infinite loop, $p_i$ starts to compute its preliminary topology view by calling $ComputeResults$
in line~\ref{ln:ComputeResults}. Then, every node $p_k$ in the queue $InformedTopology$, node $p_i$ goes over the messages in the
queue from head to bottom. While iterating the queue, for every message $m$ with source $p_k$, neighborhood $N_k$ and visited path $Path_k$,
$p_i$ inserts $Path_k$ to $opinion[N_k]$, see line~\ref{ln:insertPath}. After inserting, $p_i$ checks if there is a
neighborhood $Neig_k$ for which $opinion[Neig_k]$ contains at least $f+1$ vertex-disjoint paths,
see line~\ref{ln:found}. When such a neighborhood is found, it is stored in the $Result$
array (line~\ref{ln:found}). In line~\ref{ln:count}, $p_i$ stores the number of vertex
disjoint paths relayed messages that contained the selected neighborhood for $p_k$.
After computing an initial view of the topology, in line~\ref{ln:RemoveContradictions},
$p_i$ removes non-existing nodes from the computed topology.
For every message $m$ in $InformedTopology$, node $p_i$ aims at validating
its visited path. In line~\ref{ln:removeContradictions}, $p_i$ checks if there
exists a node $p_k$ whose neighborhood contradicts the visited path of $m$.
If such a node exists, $p_i$ decreases the associated entry in the $Count$ array
(line~\ref{ln:decreaseCount}). This decrease may cause $Count[r]$
to be smaller than $f+1$, in this case $p_i$ considers $p_k$ to be fake
and deletes the local topology of $p_k$ from $Result[r]$
(line~\ref{ln:resellgetsempty}).

%-----------------------------------------------------------------
\begin{wrapfigure}{r}{0.55\linewidth}
\BBB%\BBB
%\begin{figure*}[t!]
\fbox{
\begin{scriptsize}
\begin{minipage}{0.93\linewidth}
%\begin{superitemize}
\noindent $\bullet$ $Insert(m)$: Insert item $m$ to the queue head.

\noindent $\bullet$ $Remove(Message m)$: Remove item $m$ from the queue.

\noindent $\bullet$ $Iterator()$: Returns an pointer for iterating over the queue's items by their residence order in the queue.

\noindent $\bullet$ $HasNext()$: Tests whether the Iterator is at the queue end.

\noindent $\bullet$ $Next()$ Returns the next element to iterate over.

\noindent $\bullet$ $SizeOf()$ Returns the number of elements in the calling set.

\noindent $\bullet$ $MoveToHead(m)$: Move item $m$ to the queue head.

\noindent $\bullet$ $IsAfter(m, S)$: Test that item $m$ is after the items $m^{\prime} \in S$, where $S$ is the queue item set.
%\end{superitemize}
\end{minipage}
\end{scriptsize}
}
\caption{\B$Queue$: general purpose data structure for queuing items, and its operation list.\BBB}
\label{fig:queue}
%\end{figure*}
\BBB
\end{wrapfigure}
%-----------------------------------------------------------------------------

Upon receiving a message $m$, node $p_i$ inserts the message to the queue, in case it does not
already exist, and just moves it to the top of the queue in case it does.
The node $p_i$ now needs to relay the message $p_i$ got to all
neighbors that are not on the message visited path (line~\ref{ln:Send}).
When sending, $p_i$ also attaches the identifier of the node, from which the message was received,
to the visited path of the message.

%-----------------------------------------------------------------
\LinesNotNumbered \begin{algorithm*}[t!]
\begin{scriptsize}
\Input{$Neighborhood_{i}$: The ids of the nodes with which node $p_i$ can communicate directly}
\Output{$ConfirmedTopology \subset P \times P$: Discovered topology, which is represent by a directed edge set}

\Variable{$InformedTopology:Queue$, see Figure~\ref{fig:queue}: topological messages, $\langle node, neighborhood, path \rangle$}
\Funct{$NodeDisjointPaths(S)$: Test $S = \{ \langle node, neighborhood, path \rangle \}$ to encode at least $f+1$ vertex-disjoint paths}
\Funct{$PathContradictsNeighborhood(k, Neighborhood_k, path)$: Test that there is no node $p_j \in N$ for which there is an edge $(p_k, p_j)$ in the message's visited path, $path \subseteq P \times N$, such that $(p_k, p_j)$ is contradicting with $Neighborhood_k$}

%\Funct{$NodeDisjointMessages(S)$: Test that $S = \{ \langle node, neighborhood, path \rangle \}$ includes two messages, $m_1 = \langle p_k, Neighborhood_k, path_1\rangle$ and $m_2 = \langle p_k, Neighborhood_k, path_2\rangle$, such that $path_1$ and $path_2$ are vertex-disjoint, where $p_k \in N$ is a node and $Neighborhood_k$ is its neighborhood}

\nl \While{{\bf true}}{ \nllabel{ln:DoForever}

\nl $Result \gets ComputeResults()$  \nllabel{ln:ComputeResults}
%\tcc{Array of sets used for output calculation}
%{\bf let} $Result \gets ComputeResults()$  \nllabel{ln:ComputeResults} \tcc{Array of sets used for output calculation}

\nl {\bf let} $Result \gets RemoveContradictions(Result)$ \nllabel{ln:RemoveContradictions}

\nl $RemoveGarbage(Result)$ \nllabel{ln:RemoveGarbage}

\nl $ConfirmedTopology \gets ConfirmedTopology \cup ( \bigcup_{ p_k \in P } Result[k])$ \ \nllabel{ln:ConTopGetsConTop}

\nl \lForEach{$p_k \in N_i$}{{\bf send}$(i, Neighborhood_{i}, \emptyset)$ {\bf to } $p_k$} \nllabel{ln:sendNeighborhood}
}

%On receiving(r,R,V) from q:
\nl \Receive{($\langle \ell, Neighborhood_{\ell}, VisitedPath_{\ell} \rangle$) {\bf from} $p_j$} \nllabel{ln:Receive}

 \Begin{

\nl $Insert( p_\ell, Neighborhood_{\ell}, VisitedPath_{\ell} \cup \{ j \})$

%\tcc{Broadcast to all neighbors that have yet to received it}

\nl \lForEach{$p_k \in N_i$}{  \nllabel{ln:Send}
	   \lIf{$k \not \in VisitedPath_{\ell}$}{{\bf send}$(p_{\ell}, Neighborhood_{\ell}, VisitedPath_{\ell} \cup \{ j \}$) {\bf to } $p_k$ } \nllabel{c:VisitedPathUpj}
    }
}

\nl \Proced{$Insert(k, Neighborhood_k, VisitedPath_k)$}

 \Begin{

%\tcc{If message is already in $InformedTopology$, move it. Else if the message is valid, insert it}\

\nl \lIf{$ \langle k, Neighborhood_{k}, VisitedPath_{k} \rangle \in InformedTopology$}{$InformedTopology.MoveToHead(m)$}

\nl \lElseIf{$p_k \in N \wedge Neighborhood_k \subseteq indices(N) \wedge VisitedPath_k  \subseteq indices(N)$}{$InformedTopology.Insert(\langle k, Neighborhood_k, VisitedPath_k \rangle)$}

}

\nl \Funct{$ComputeResults()$}
 \Begin{

\nl \ForEach{$p_k \in P : \langle k, Neighborhood_k, VisitedPath_k \rangle \in InformedTopology$}{

    \nl {\bf let} $(FirstDisjointPathsFound, Message, opinion[]) \gets (false, InformedTopology.Iterator(), [\emptyset])$\

    \nl \While{$Message.hasNext()$}{
        \nl $\langle \ell, Neighborhood_{\ell}, VisitedPath_{\ell} \rangle  \gets Message.Next()$\

     \nl \lIf{$\ell = k$}{$opinion[Neighborhood_{\ell}].Insert(\langle$ $\ell,$ $Neighborhood_{\ell},$ $VisitedPath_{\ell}\rangle)$} \nllabel{ln:insertPath}

        \nl \If{$FirstDisjointPathsFound = false$ $\wedge$ $NodeDisjointPaths(opinion[Neighborhood_{\ell}])$}{$(Result[k],$ $FirstDisjointPathsFound) \gets (Neighborhood_{\ell},$ ${\bf true})$} \nllabel{ln:found}

        }

	\nl $Count[k] \gets opinion[k][Result[k.SizeOf()$ \nllabel{ln:count}
    }
\nl \Return $Result$
}

\nl \Funct{$RemoveContradictions(Result)$}
 \Begin{

    \nl \ForEach{$\langle r, Neighborhood_r, VisitedPath_r \rangle \in InformedTopology$}{

        \nl \If{$\exists p_k \in P : PathContradictsNeighborhood(p_k, Result[k], VisitedPath_r)={\bf true}$}{  \nllabel{ln:removeContradictions}

            \nl \lIf{$Neighborhood_r=Result[r]$}{$Count[r] \gets Count[r] - 1$\ \nllabel{ln:decreaseCount}}

            \nl \lIf{$Count[r] \leq f$}{$Result[r] \gets \emptyset$}\ \nllabel{ln:resellgetsempty}
        }
    }
\nl \Return $Result$
}

%%%% Remove edge loops
\nl \Proced{$RemoveGarbage(Result)$}
 \Begin{
    \nl \ForEach{$p_k \in N$}{

        \nl \lForEach{$m=\langle k, Neighborhood_k, VisitedPath_k \rangle \in InformedTopology : \{ k \} \cup  Neighborhood_k  \cup VisitedPath_k \not \subseteq P$ $\vee$ $InformedTopology.IsAfter(m,$ $opinion[k][Result[k]$$])$ }{
             $InformedTopology.Remove(m)$
        }
    }
}
\end{scriptsize}
\caption{Topology discovery (code for node $p_i$)}

\label{algo:discovery}
\end{algorithm*}

\Subsection{Algorithm's correctness proof.}
We now prove that within a linear amount of asynchronous rounds, the system stabilizes and every output is legal.
The proof considers an arbitrary starting configuration with arbitrary messages in transit that could be actually in the communication channel or already stored in $p_j$'s message queue and will be forwarded in the next steps of $p_j$. Each message in transit that traverse correct nodes can be forwarded within less than $\BigO(|C|)$ asynchronous rounds.
Note that any message that traverses Byzantine nodes and arrives to a correct node that has at least one Byzantine node in its path. The reason is that the correct neighbor to the last Byzantine in the path lists the Byzantine node when forwarding the message.
Thus, $f$ is at most the number of messages that encode vertex-disjoint paths from a certain source that are initiated or corrupted by a Byzantine node. Since there are at least $f+1$ vertex-disjoint paths with no Byzantine nodes from any source $p_k$ to any node $p_i$ and since $p_k$ repeatedly sends messages to all nodes on all possible paths, $p_i$ receives at least $f+1$ messages from $p_k$ with vertex-disjoint paths.

The FIFO queue usage and the repeated send operations of $p_k$ ensure that the most recent $f+1$ messages with vertex-disjoint paths in $InformedTopology$ queue are uncorrupted messages. Namely, misleading messages that were present in the initial configuration will be pushed to appear below the new $f+1$ uncorrupted messages. Thus, each node $p_i$ eventually has the local topology of each correct node (stored in the $Result_i$ array). The opposite is however not correct as local topologies of non-existing nodes may still appear in the result array. For example, $InformedTopology_i$ may include in the first configuration $f+1$ messages with vertex-disjoint paths for a non-existing node.
Since after $ComputeResults$ we know the correct neighborhood of each correct node $p_k$, we may try to ensure the validity of all messages.
For every message that encodes a non-existing source node, there must be a node $p_\ell$ on the message path, such that $p_\ell$ is correct and
$p_\ell$'s neighbor is non-existing, this is true since $p_i$ itself is correct. Thus, we may identify these messages and ignore them. Furthermore, no valid messages are ignored because of this validity check.

We also note that, since we assume that the nodes of the system are a subset of $P$, the size of the queue $InformedTopology$ is bounded. Lemma~\ref{l:BoundedMemory} bounds the needed amount of node memory (the proof details appear in Section~\ref{s:discoveryApp} of the Appendix). % and in~\cite{TR}. %After showing correctness and convergence time, .

\BB \begin{lemma}[Bounded memory]
\label{l:BoundedMemory}
%Let $p_i \in C$ be a correct node.
At any time, there are at most $n \cdot 2^{2 n}$ messages in $InformedTopology_i$, where $p_i \in C$, $n = |P|$ and $\BigO(n \log (n))$ is the message size.
\BB \end{lemma}

\Subsection{$r$-neighborhood discovery.}
Algorithm~\ref{algo:discovery} demonstrates the existence of a deterministic self-stabilizing Byzantine resilient algorithm for topology discovery. Lemma~\ref{l:BoundedMemory} shows that the memory costs are high when the entire system topology is to be discovered. We note that one may consider the task of $r$-neighborhood discovery. Recall that in the $r$-neighborhood discovery task, it is assumed that every $r$-neighborhood cannot be partitioned by Byzantine nodes. Therefore, it is sufficient to constrain the maximal path length in line~\ref{ln:Send}. The correctness proof of the algorithm for the $r$-neighborhood discovery follows similar arguments to the correctness proof of Algorithm~\ref{algo:discovery}.
%
%In detail, Claim~\ref{c:CorrectRelaying} and Lemma~\ref{l:ClearChannel} use inductive arguments for demonstrating their correctness for path of any size, in particular $r$.

\Section{End-to-End Delivery} %Self-Stabilizing Byzantine
\label{s:e2e}
We present a design for a self-stabilizing Byzantine resilient algorithm for the transport layer protocol that uses the output of Algorithm~\ref{algo:discovery}.  The design is based on a function (named $getDisjointPaths()$) for selecting vertex-disjoint paths that contain a set of $f+1$ correct vertex-disjoint paths. We use $getDisjointPaths()$ and ARQ (Automatic Repeat reQuest) techniques for designing Algorithm~\ref{algo:end2end}, which ensures safe delivery between sender and receiver.

\Subsection{Exchanging messages over $f+1$ correct vertex-disjoint paths}
%The procedure $ByzantineFaultTolerantSend$ uses the output of Algorithm~\ref{algo:discovery}, which repeatedly discovers the network topology.
%The sender and the receiver exchange messages over $f+1$ correct vertex-disjoint paths.
We guarantee correct message exchange by sending messages over a polynomial number of vertex-disjoint paths between the sender and the receiver.
We consider a set, $CorrectPaths$, that includes $f+1$ correct vertex-disjoint paths. Suppose that $ConfirmedTopology$ (see the output of Algorithm~\ref{algo:discovery}) encodes a set, $Paths$, of $2f+1$ vertex-disjoint paths between the sender and the receiver. It can be shown that
$Paths$ includes at most $f$ incorrect paths that each contain at least one Byzantine node, i.e., $Paths \supseteq CorrectPaths$. As we see next, $ConfirmedTopology$ does not always encode $Paths$, thus, one needs to circumvent this difficultly.

Note that even though $2f+1$ vertex-disjoint paths between the sender and the receiver are present in the communication graph, the discovered topology in $ConfirmedTopology$ may not encode the set $Paths$, because $f$ of the paths in the set $Paths$ can be controlled by Byzantine nodes. Namely, the information about at least one edge in $f$ of the paths in the set $Paths$, can be missing in $ConfirmedTopology$.

We consider the problem of relaying messages over the set $CorrectPaths$ when only $ConfirmedTopology$ is known, and propose three implementations to the function $getDisjointPaths()$ in Figure~\ref{fig:procedures}. The value of $ConfirmedTopology$ is a set of directed edges $(p_i,p_j)$. An undirected edge is approved if both $(p_i,p_j)$ and $(p_j,p_i)$ appear in $ConfirmedTopology$. Other edges in $ConfirmedTopology$ are said to be suspicious.
%
%The arguments used in Figure~\ref{fig:procedures} assume that the system is in a safe configuration with respect to Algorithm~\ref{algo:discovery}.
For each of the proposed implementations, we show in Section~\ref{s:ipf} of the Appendix that a polynomial number of paths are used and that they contain $CorrectPaths$. Thus, the sender and the receiver can exchange messages using a polynomial number of paths and message send operations, because each path is of linear length.

%We focus on a solution suitable for the general case, which uses polynomial time and message cost.

%The solution considers Definition~\ref{d:SuspiciousEdge} and assumes that Byzantine nodes cannot be immediate neighbors and that all neighbors of a given Byzantine node refer to the Byzantine with the same identifier.

%If there exists at least one such path set, the sender can safely use them to communicate with the receiver (similar to Algorithm~\ref{algo:discovery}). However, the discovered topology may not include even one such set. The reason is that $f$ of the paths that should appear in the discovered topology may be controlled by Byzantine nodes. [@@ Didn't assume that the graph include such paths? It is confusing. @@ Namely, the information about at least one edge in each such path may not arrive to the sender.

\begin{scriptsize}
%-----------------------------------------------------------------
\begin{figure*}[t!]%{r}{0.47\linewidth}
%\B%\B%BB
%\begin{figure*}[t!]
\fbox{
\begin{small}
\begin{minipage}{0.955\linewidth}
%\begin{superitemize}
{\bf The case of constant $r$ and $\Delta$.}~~~~ The sender and the receiver exchange messages by using all possible paths between them; feasible considering $r$-neighborhoods, where the neighborhood radius, $r$, and the node degree $\Delta$ are constants.

{\bf The case of constant $f$.} ~~~~
%This procedure entails sending a message on a path set, $Paths$, where $|Paths|$ is polynomial and $CorrectPaths \subseteq Paths$. %Namely, Moreover, these paths are sufficient to guarantee safe delivery of the message.
%
%We explain how the sender and the receiver select a set of vertex-disjoint paths, ${\cal P}(p_1,p_2, \ldots p_f) \subseteq Paths$, that contains $f+1$ correct vertex-disjoint paths.
%
%${\cal P}(p_1,p_2, \ldots p_f) \subseteq Paths$
%
For each possible choice of $f$ system nodes, $p_1,p_2, \ldots p_f$, the sender and the reciter compute a new graph $G(p_1,p_2, \ldots p_f)$ that is the result of removing $p_1,p_2, \ldots p_f$, from $G_{out}$, which is the graph defined by the discovered topology, $ConfirmedTopology$. Let ${\cal P}(p_1,p_2, \ldots p_f)$ be a set of $f+1$ vertex-disjoint paths in $G(p_1,p_2, \ldots p_f)$ (or the empty set when ${\cal P}(p_1,p_2, \ldots p_f)$ does not exists) and $Paths = \bigcup_{p_1,p_2, \ldots p_f} {\cal P}(p_1,p_2, \ldots p_f)$. The sender and the receiver can exchange messages over $Paths$, because $|Paths|$ is polynomial at least one choice of $p_1,p_2, \ldots p_f$, has a corresponding set ${\cal P}(p_1,p_2, \ldots p_f)$ that contains $CorrectPaths$ (Section~\ref{s:ipf} of the Appendix).

%First we show that this procedure only sends messages through a polynomial number of paths. There are $\BigO(n^f)$ possible chooses of $f$ system nodes. Thus, $\BigO(n^f )$ path sets are computed, and since $f$ is a constant, this number is polynomial. Moreover, each such set contains at most $f + 1$ simple paths of linear length, because $p_i$ only computes sets, ${\cal P}(p_1,p_2, \ldots p_f)$, of size $f+1$. Thus, the sender and the receiver can exchange messages using a polynomial number of paths and message send operations.

%We show that $CorrectPaths \subseteq Paths$. Consider the permutation choice, $p_1,p_2, \ldots p_f$, in which the set actually contains the set of Byzantine nodes in the system. Thus $G(p_1,p_2, \ldots p_f)$ contains only correct nodes. Furthermore, at least $f+1$ paths that were present in $G_{out}$ are still present in $G(p_1,p_2, \ldots p_f)$, since $G(p_1,p_2, \ldots p_f)$ was obtained from $G_{out}$ by removing of $f$ (Byzantine) nodes, $p_1,p_2, \ldots p_f$. Hence, there are at least $f+1$ correct vertex-disjoint paths in $G(p_1,p_2, \ldots p_f)$, in ${\cal P}(p_1,p_2, \ldots p_f)$ and in $Paths$.

{\bf The case of no Byzantine neighbors} ~~~~ The procedure assumes that any Byzantine node has no directly connected Byzantine neighbor in the communication graph. Specifically, this polynomial cost solution considers the (extended) graph, $G_{ext}$, that includes all the edges in $confirmedTopology$ and {\em suspicious edges}. Given three nodes, $p_i, p_j, p_k \in P$, we say that node $p_i$ considers the undirected edge $(p_k, p_j)$ suspicious, if the edge appears as a directed edge in $ConfirmedTopology_i$ for only one direction, e.g., $(p_j,p_k)$.

%, see Definition~\ref{d:SuspiciousEdge}.

%\BB\B
%\begin{definition}[Suspicious edges]
%%
%\label{d:SuspiciousEdge}
%%
%Given three nodes, $p_i, p_j, p_k \in P$, we say that node $p_i$ considers the undirected edge $(p_k, p_j)$ suspicious, if the edge appears as a directed edge in $ConfirmedTopology_i$ for only one direction, e.g., $(p_j,p_k)$.
%\end{definition}
%\BB\B

The extended graph, $G_{ext}$, may contain fake edges that do not exists in the communication graph, but Byzantine nodes reports on their existence. Nevertheless, $G_{ext}$ includes all the correct paths of the communication graph, $G$. Therefore, the $2f+1$ vertex-disjoint paths that exists in $G$ also exists in $G_{ext}$ and they can facilitate a polynomial cost solution for the message exchange task (Section~\ref{s:ipf} of the Appendix).

%Let $G^{\prime}=(N, E_{G^{\prime}})$ be the graph computed from $ConfirmedTopology$ and its suspicious edges. %, see Definition~\ref{d:SuspiciousEdge}. We demonstrate that $G^{\prime}$'s edges, $E_{G^{\prime}}$, contains the edges, $E_{G}$, of the communication graph, $G$.
%
%Let us consider $e=(p_j, p_k) \in E_{G}$ and show that $e \in E_{G^{\prime}}$. When both $p_j$ and $p_k$ are correct, the correctness of Algorithm~\ref{algo:discovery} implies $e \in E_{G^{\prime}}$. Suppose that $p_j$ is correct and $p_k$ is Byzantine, and consider the different cases in which $p_k$ decides to report (or not to report) about $e$ as part of its local neighborhood. Namely, either $e \in ConfirmedTopology$, or $e$ is a suspicious edge, because $p_i$ reports about $e$, and $p_k$ decides to report, and respectively, not to report. Since $G \subseteq G^{\prime}$, $G^{\prime}$ must contain $2f+1$ vertex-disjoint paths between any sender $p_s$ and receiver $p_r$, because  $G$ does.
%
%Moreover, the same arguments implies that there may be at most $f$ incorrect paths, which contain each at least one Byzantine node. Hence, there are at least $f+1$ correct vertex-disjoint paths in $Paths$.

%\end{superitemize}
\end{minipage}
\end{small}
}
\caption{\B Implementation proposals for the function $getDisjointPaths()$.\BBB}
\label{fig:procedures}
%\end{figure*}

\end{figure*}
%-----------------------------------------------------------------------------
\end{scriptsize}

\Subsection{Ensuring safe message delivery}
We propose a way for the sender and the receiver, that exchange a message over the paths in $getDisjointPaths()$, to stop considering messages and acknowledgments sent by Byzantine nodes. They repeatedly send messages and acknowledgments over the selected vertex-disjoint paths. Before message or acknowledgment delivery, the sender and the receiver expect to receive each message and acknowledgment at least $(capacity \cdot n+1)$ consecutive times over at least $f+1$ vertex independent paths, and by that provide evidence that their messages and acknowledgments were indeed sent by them.

We employ techniques for labeling the messages (in an ARQ style), recording visited path of each message, and counting the number of received message over each path. The sender sends messages to the receiver, and the receiver responds with acknowledgments after these messages are delivered to the application layer. Once the sender receives the acknowledgment, it can fetch the next message that should be sent to the receiver. The difficulty here is to guarantee that the sender and receiver can indeed exchange messages and acknowledgments between them, and stop considering messages and acknowledgments sent by Byzantine nodes.

The sender repeatedly sends message $m$, which is identified by $m.ARQLabel$, to the receiver over all selected paths. The sender does not stop sending $m$ before it is guaranteed that $m$ was delivered to the application layer of the receiving-side. When the receiver receives the message, the set $m.VisitedPath$ encodes the path along which $m$ was relayed over. Before delivery, the receiver expects to receive $m$ at least $(capacity \cdot n+1)$ consecutive times from at least $f+1$ vertex independent paths. Waiting for $(capacity \cdot n+1)$ consecutive messages on each path, ensures that the receiver gets at least one message which was actually sent recently by the sender. Once the receiver delivers $m$ to the application layer, the receiver starts to repeatedly acknowledge with the label $m.ARQLabel$ over the selected paths (while recording the visited path). The sender expects to receive $m$'s acknowledgment at least $capacity \cdot n+1$ consecutive times from at least $f+1$ vertex independent paths before concluding that $m$ was delivered to the application layer of the receiving-side.

Once the receiver delivers a message to the application layer, the receiver starts to repeatedly acknowledge the recently delivered message over the selected paths. In addition, the receiver also restarts its counters and the log of received messages upon a message delivery to the application layer.
Similarly the sender count acknowledgments to the current label used, when the sender receives at least $capacity \cdot n+1$ acknowledgments over $f+1$ vertex-disjoint paths, the sender fetches the next message from the application layer, changes the label and starts to send the new message.

%We note that starting from an arbitrary configuration, the sender eventually fetches a message from the application layer. This is obvious since if the sender is sending the same message forever, then the receiver counters on $f+1$ paths must exceed $capacity \cdot n+1$. From this point the receiver sends acknowledgments with the correct label forever ensuring that the sender fetches the next message.

%-----------------------------------------------------------------
\begin{scriptsize}
\LinesNotNumbered \begin{algorithm*}[t!]
\begin{scriptsize}
\Inter{$FetchMessage()$: Gets messages from the upper layer. We denote by $InputMessageQueue$ the unbounded queue of all messages that are to be delivered to the destination}
\Inter{$DeliverMessage(Source, Message)$: Deliver an arriving message to the higher layer. We denote by $OutputMessageQueue$ the unbounded queue of all  messages that are to be delivered to the higher layer. We assume that it always contains at least the last message inserted to it}
\Input{$ConfirmedTopology$: The discovered topology (represented by a directed edge set, see Algorithm~\ref{algo:discovery})}
\DataStructure{Transport layer messages: $\langle Source,$ $Destination,$ $VisitedPath,$ $IntentedPath,$ $ARQLabel,$ $Type,$ $Payload \rangle$, where $Source$ is the sending node, $Destination$ is the target node, $VisitedPath$ is the actual relay path, $IntentedPath$ is the planned relay path, $ARQLabel$ is the sequence number of the stop-and-wait ARQ protocol, and $Type$ $\in$ $\{ Data,$ $ACK \}$ message type, where DATA and ACK are constant}
%~~~~~\Field{$Payload$: the message data}
%\DataStructure{$Queue$: General purpose data structure for queuing items}
%~~~~~\Oper{$Insert(m)$: Insert item $m$ to the head of the queue}
\Variable{$Message$: the current message being sent}
\Variable{$ReceivedMessages[j][Path]$ : queue of $p_j$'s messages that were relayed over path $Path$}
\Variable{$Confirmations[j][Path]$ : $p_j$'s acknowledgment queue for messages that were relayed over $Path$}
\Variable{$label$: the current sequence number of the stop-and-wait ARQ protocol}
\Variable{$Approved$: A Boolean variable indicating whether $Message$ was accepted at the destination}
\Funct{$NodeDisjointPaths(S)$: Test $S$, a set of paths, to encode at least $f+1$ vertex-disjoint paths}
\Funct{$FloodedPath(MessageQueue, m)$ : Test whether $m$ is encoded by the first $capacity \cdot n + 1$ messages in $MessageQueue$.}
%, where $capacity$ is an upper bound on the number of messages in transit over a communication link.}
%\Funct{$SuspiciousEdges()$ : Get the set of suspicious edges}
\Funct{$getDisjointPaths(ReportedTopology, Source, Destination)$ : Get a set of vertex-disjoint paths between $Source$ and $Destination$ in the discovered graph, $ReportedTopology$ (Figure~\ref{fig:procedures}).}
\Funct{$ClearQueue(Source)$ : Delete all data in $ReceivedMessages[Source][\ast]$}
\Funct{$ClearAckQueue(Destination)$ : Delete all data in $Confirmations[Destination][\ast]$}

\nl \While{{\bf true}}{
\nl $ClearAckQueue(Message.Destination)$\\   \nllabel{ln:ClearAcks}
\nl  $(Message, label) \gets (FetchMessage(), {label} + 1 ~modulo ~3)$\\         \nllabel{ln:fetch}  \nllabel{ln:changeLabel}
\nl  \lWhile{$Approved = {\bf false}$}{$ByzantineFaultTolerantSend(Message)$}  \nllabel{ln:FaultSend}
}
\nl  \Receive{($msg$) {\bf From} $p_j$}   \nllabel{ln:receive1}
\Begin{
\nl  \If{$msg.Destination \neq i$}{   \nllabel{ln:checkDest}
\nl      $msg.VisitedPath \gets msg.VisitedPath \cup \{ j \}$  \nllabel{ln:attach}

\nl      ${\bf send}$($msg$) {\bf to}  {\bf next} ($msg.IntendedPath$) \nllabel{ln:notdestsend}
}

\nl  \ElseIf{$msg.Type = Data$}{ \nllabel{ln:checkType}
\nl  		 $ReceivedMessages[msg.Source][msg.VisitedPath].insert(\langle$ $msg.Payload,$ $msg.ARQLabel$ $\rangle)$    \nllabel{ln:insertData}

\nl           {\bf let} $Paths \gets \{ Path$ $:$ $FloodedPath(Confirmations[msg.Source][Path],$ $msg)\}$
		
\nl  		\If{$NodeDisjointPaths(Paths)$}{ \nllabel{ln:isEnough}

\nl  $NewMesssage \gets {\bf true}$ \nllabel{ln:newMessageTrue}

\nl  $Confirm(msg.Source, m.ARQLabel, m.Payload)$   \nllabel{ln:confirmMessage}
}

	}
\nl  \ElseIf{$msg.Type = ACK$}{
\nl        \lIf{$label = msg.ARQLabel$}{$Confirmations[msg.Source][msg.VisitedPath].insert(\langle msg.Payload, msg.ARQLabel \rangle)$}

\nl           {\bf let} $Paths \gets \{ Path$ $:$ $FloodedPath(Confirmations[msg.Source][Path],$ $\langle msg.Payload,$ $msg.ARQLabel$ $\rangle )\}$ \nllabel{ln:DataApprovedTrue}
		
\nl  		\lIf{$NodeDisjointPaths(Paths)$}{$Approved \gets {\bf true}$} \nllabel{ln:ApprovedTrue} \nllabel{ln:isEnoughConf}

}
}

\nl  \Funct{$Confirm(Source, ARQLabel, Payload)$}
\Begin{
\nl     	\lIf{$CurrentLabel \neq ARQLabel$}{$DeliverMessage(Source, Payload)$} \nllabel{ln:deleiverMessage} %@@new@@

\nl  	$(CurrentLabel, NewMessage) \gets (ARQLabel, {\bf false})$  \nllabel{ln:AssignLabel} \nllabel{ln:newMessageFalse}

\nl  	$ClearQueue(Source)$ \nllabel{ln:clear}

%\nl  	$NewMessage \gets {\bf false}$\\\nllabel{ln:newMessageFalse}
%\nl  	$NewMessage = {\bf false}$\\\nllabel{ln:newMessageFalse} Tell Omri to fix.

\nl  	\lWhile{$NewMessage = {\bf false}$}{$ByzantineFaultTolerantSend(\langle$ $Source,$ $ARQLabel,$ $ACK,$ $Payload \rangle)$} \nllabel{ln:sendAcks}
}

\nl  \Funct{$ByzantineFaultTolerantSend(\langle Destination, ARQLabel, Type, Payload \rangle)$}
\Begin{

\nl  	{\bf let} $Paths \gets getDisjointPaths(ConfirmedTopology, i, Destination)$ \nllabel{ln:getDisjoint}

% \cup SuspiciousEdges()

\nl  	\lForEach{$Path \in Paths$}{
          {\bf send}$(\langle i, Destination, \emptyset, Path, ARQLabel, Type,           Payload \rangle)$ {\bf to} {\bf first}$(Path)$ \nllabel{ln:send}
     }
	
}
\end{scriptsize}
\caption{Self-stabilizing Byzantine resilient end-to-end delivery ($p_i$'s code)}

\label{algo:end2end}
\end{algorithm*}
\end{scriptsize}
%-----------------------------------------------------------------------------

\Subsection{The pseudocode of Algorithm~\ref{algo:end2end}}
In every iteration of the infinite loop, $p_i$ fetches $Message$, prepares $Message$'s label (line ~\ref{ln:changeLabel}) and starts sending $Message$ over the selected paths, see the procedure $ByzantineFaultTolerantSend(Message)$. When $p_i$ gets enough acknowledgments for $Message$ (line~\ref{ln:FaultSend}), $p_i$ stops sending the current message and fetches the next.
Upon receiving a message $msg$, node $p_i$ tests $msg$'s destination (line~\ref{ln:checkDest}). When $p_i$ is not $msg$'s destination, it forwards $msg$ to the next node on $msg$'s intended path, after updating $msg$'s visited path. When $p_i$ is $msg$'s destination,
$p_i$ checks $msg$'s type (line~\ref{ln:checkType}). When $msg$'s type is {\em Data}, $p_i$ inserts the message payload and label
to the part of the data structure associated with the message source, i.e., the sender,
and the message visited path (line~\ref{ln:insertData}). In line~\ref{ln:isEnough}, node $p_i$ checks whether $f+1$ vertex-disjoint paths
relayed the message at least $capacity \cdot n +1 $ times, where $capacity$ is an
upper bound on the number of messages in transit over a communication link.
If so, $p_i$ delivers the $msg$ to the application layer (line~\ref{ln:deleiverMessage}), clears the
entire data structure and finally sends acknowledgments on the selected paths
until a new message is confirmed. Moreover, in line~\ref{ln:newMessageFalse} we signal that
we are ready to receive a new message.
When $msg$'s type is $ACK$, we act almost as when the message is of type $Data$. When the condition in line~\ref{ln:ApprovedTrue} holds, we signal that the message was confirmed at the receiver by setting $Approved$ to be $true$, in line~\ref{ln:ApprovedTrue}.
We note that the code of Algorithm~\ref{algo:end2end} considers only one possible pair of source and destination. A many-source to many-destination version of this algorithm can simply use a separate instantiation of this algorithm for each pair of source and destination.

%@@ Define legal execution.
%We have an unbounded Message input Queue and we have  message output queue that always includes at least the last message added to it.
%In a legal execution messages from the sender message input queue are fetched, send an enter the receiver's message output queue by their sending order.
%@@

\Subsection{Correctness proof.}
We show that message delivery guarantees hold after a bounded convergence period. The proof is based on the system's ability to relay messages over $f+1$ correct vertex-disjoint messages (Figure~\ref{fig:procedures}), and focuses on showing safe message delivery between the sender and the receiver. After proving that the sender fetches messages infinitely often, we show that within four such fetches, the message delivery guarantees hold; receiver-side delivers all of the sender's messages and just them. The proof in detail appears in Section~\ref{s:end2endproof} of the Appendix. % and in~\cite{TR}.

%The proof arguments are based on the ARQ labels that manage the correct message delivery.
%
%Following the fetch of each of the first three messages and before the next one, the sender must count $capacity \cdot n+1$ acknowledgments with the current label that the sender uses to send, namely with $0$, $1$ and $2$. Since the sender reset the counters when changing the sending label to $1$, the receiver must send at least one acknowledgment with label $1$ and then with label $2$, following the corresponding fetches. Thus, the receiver must clear its counters at least once following the second fetch and before the fourth fetch and then start sending acknowledgments with label $2$. After clearing the counters by the receiver and starting sending acknowledgments with label $2$ a message with label $0$ that is next to be sent, must be delivered and no other message can be counted as arriving at least $capacity \cdot n+1$ times through $f+1$ vertex-disjoint paths.

%
%Let us consider three labels, $0$, $1$, and $2$ that are used by the sender in a round robin fashion.
%
%
Let us consider messages, $m$, and their acknowledgements, that arrive at least $(capacity \cdot n+1)$ times over $f+1$ vertex-independent paths, to the receiver-side, and respectively the sender-side, with identical payloads and labels. The receiver, and respectively the sender,
has the {\em evidence} that $m$ was {\em indeed sent by the sender,} and respectively, {\em acknowledged by the receiver.} The sender and the receiver {\em clear} their {\em logs} whenever they have such evidences about $m$. The proof shows that, after a finite convergence period, the system reaches an execution in which the following events reoccur: (\op{Fetch}) the sender clears its log, fetches message $m$, and sends it to the receiver, (\op{R-Get}) the receiver gets the evidence that $m$ was indeed sent by the sender, (\op{Deliver}) the receiver clears its log, delivers $m$, and acknowledge it  to the sender, and (\op{S-Get}) the sender gets the evidence that $m$ was acknowledged by the receiver. Namely, the system reaches a legal execution.

%Whenever these $capacity \cdot n+1$ messages have the identical payloads and labels, the receiver delivers them, nullify the counters, empty queues and send acknowledges with the label of the delivered message over a selected set of vertex-disjoint paths (cf. line~\ref{ln:confirmMessage}).
%
%
%The sender clears counters and queues whenever the sender changes label.

First we prove that event \op{Fetch} occurs infinitely often, in the way of proof by contradiction. % (Lemma~\ref{l:RfeAfell}).
Let us assume (towards a contradiction) that the sender fetches message $m$ and then never fetches another message $m^\prime$. The sender sends $m$ and counts acknowledgments that has $m$'s label. According to the algorithm, the sender can fetch the next message, $m^\prime \neq m$, when it has the evidence that $m$ was indeed acknowledged by the receiver. The receiver acknowledges $m$'s reception when it has the evidence that $m$ was indeed sent by the sender. After nullifying its logs, the receiver repeatedly sends $m$'s acknowledgments until it has evidences for other messages, $m^\prime$, that were indeed sent by the sender after $m$. By the assumption that the sender never fetches $m^\prime \neq m$, we have that the receiver keeps on acknowledging $m$ until $m^\prime \neq m$ arrives from the sender. Therefore, $m$ arrives from the sender to the receiver, and the receiver acknowledges $m$ to the sender. Thus, a contradiction that the sender never fetches $m^\prime \neq m$.

The rest of the proof shows that (eventually) between every two event of type \op{Fetch}, also the events \op{R-Get}, \op{Deliver} and \op{S-Get} occur (and in that order). We show that this is guaranteed within four occurrences of event \op{Fetch}.
%
%[label_m -> link_{s,r} -> label(log_r)  -> label{ack}  -> link_{r,s}  -> label(log_s)]
%
%[0 -> link_{s,r} -> 0, [1/2], 0  -> 0, [1/2], 0  -> link_{r,s}  -> [1,2], 0 ]
%
%We then turn to show that the receiver-side delivers all (and only) the sender's messages. The difficulty here is to show that, in addition to the corrupted messages that can be present in the starting configuration, the system has to deal with Byzantine nodes that can introduce corrupted messages at any point of the system execution.
%
%The proof counts the number of time that the sender fetches a new message. It shows that within four such fetches, the receiver-side delivers no Byzantine message and delivers all the sender's messages. Immediately after the first fetch, say when the sender label was change from $0$ to $1$, the sender's log is empty, thus, it cannot accept
%
%Then in between the second and the fourth fetch of any four successive fetches, where without the loss of generality, the first fetch is with label $0$, the second with $1$, the third with label $2$ and the fourth with $0$ the receiver clears its counter and the last fetched message in this sequence that is with label $0$ is later delivered.
%
%
Following the fetch of each of the first three messages and before the next one, the sender must have evidence that the receiver executed event \op{Deliver}, i.e., clearing the receiver's log. Note that during convergence, this may surely be false evidence. Just before fetching a new message in event \op{Fetch}, the sender must clear its logs and reassign a label value, say, the value is $0$. There must be a subsequent fetch with label $1$, because, as explained above, event \op{Fetch} occurs (infinitely often). Since the sender clears its logs in event \op{Fetch}, from now on and until the next event \op{Fetch}, any corrupted message found in the sender's log  must be of Byzantine origin. Therefore, the next time sender gets the evidence that $m$ was acknowledged by the receiver, the receiver has truly done so.
%We note that there are $f$ Byzantine nodes, and $2f+1$ vertex independent paths between the recover and the sender. Moreover, before the next event (1) the sender checks whether $f+1$ vertex-disjoint paths relayed an acknowledgment at least $capacity \cdot n +1 $ times. Therefore, at least one of these acknowledgments must have been sent by the receiver.
% Thus, when the sender has enough evidence regarding the message with label $1$, it must be that the receiver sent at least one acknowledgment with label $1$.
%Following the same arguments, we get that the receiver also sent an acknowledgment with label $2$.
Note that between any such two (truthful) acknowledgments (with different labels), say with label, $1, 2, \ldots$, the receiver must execute event \op{Deliver} and clean its log, see Algorithm~\ref{algo:end2end}, line~\ref{ln:clear}.
%
%Therefore, when the sender executes event \op{Fetch}, this time with label $0$, the receiver's log is clear, we have that eventually, the receiver's logs will contain evidence for the message with label $0$.
%
Since the sender sends over $f + 1$ correct paths, and the receiver's logs are clear, eventually the
receiver will have evidence for the message with label $0$. As corrupted messages
originate only from Byzantine nodes and there are at most $f$ such nodes,
the receiver's log may not contain evidence for non-sender messages.
To conclude, starting from the $4$-th message, the receiver will confirm all of the sender's messages, and will not confirm non-sender messages.

\Section{Extensions and Conclusions} %Concluding Remarks}
\label{s:ext}
%
%This work had presented two deterministic, self-stabilizing Byzantine-resilience algorithms for
%topology discovery and end-to-end message delivery. We have also considered an $r$-neighborhood
%network discovery algorithm as well.
%
As an extension to this work, we suggest to combine the algorithms for $r$-neighborhood network discovery and the
end-to-end capabilities in order to allow the use of end-to-end message delivery within
the $r$-neighborhoods. These two algorithms can be used by the nodes, under reasonable node density assumptions,
for discovering their $r$-neighborhoods, and, subsequently, extending the scope of their
end-to-end capabilities beyond their $r$-neighborhood, as we describe in the following.
We instruct further remote nodes to relay topology information, and in this way collect information on remote neighborhoods. One can consider an algorithm for studying specific remote neighborhoods that are defined, for example, by their geographic region, assuming the usage of GPS inputs; a specific direction and distance from the topology exploring node defines the exploration goal.
The algorithm nominates $2f+1$ nodes in the specific direction to return further information towards the desired direction. The sender uses end-to-end
communication to the current $2f+1$ nodes in the {\em front} of the current exploration, asks them for their $r$-neighborhood, and chooses a new set of $2f+1$ nodes for forming a new front. It then instructs each of the current nodes in the current front to communicate with each node in the chosen new front, to nominate the new front nodes
to form the exploration front.

To ensure stabilization, this interactive process of remote information collection should never stop. Whenever the current collection process investigates beyond the closest $r$-neighborhood, we concurrently start a new collection process in a pipeline fashion. The output is the result of the last finalized collection process. Thus, having a correct output after the first time a complete topology investigation is finalized.

%The studied algorithms are deterministic, and as such, offer a high predictability degree. One can consider an alternative probabilistic approach in order to circumvent the high message cost of Algorithm~\ref{algo:discovery}. Namely, when neighborhood radius, $r$, is in the order of the graph diameter, the message cost is exponential. An random algorithm can let each topology message contain a message sequence number that the sender, $p_s$, increases. The relaying nodes, $p_i$, forward merely a fixed random subset of the topological messages sent from $p_s$ and containing the same message sequence number. Each topology message is sent $\BigO(1)$ times on each edge. Thus, the message cost is polynomial. Moreover, when the receiver, $p_r$, considers topological messages, $m$, it checks to see whether a majority of vertex-disjoint paths starting in $p_s$ relayed messages that support the topological knowledge contained in $m$. Note that there is no demand that $f+1$ vertex-disjoint path relay the same topological information.  @@ Please provide more details on how and why these extensions work. @@

%In this work we presented two deterministic, self-stabilizing Byzantine-resilience algorithms for topology discovery and end-to-end message delivery. We have also considered an algorithm for discovering $r$-neighborhood in polynomial time, communication and space. Lastly, we mentioned a possible extension for exploring and communicating with remote $r$-neighborhoods using polynomial resources as well.

In this work we presented two deterministic, self-stabilizing Byzantine-resilience algorithms for topology discovery and end-to-end message delivery. We have also considered an algorithm for discovering $r$-neighborhood in polynomial time, communication and space. Lastly, we mentioned a possible extension for exploring and communicating with remote $r$-neighborhoods using polynomial resources as well.

The obtained end-to-end capabilities can be used for communicating the public keys of parties and establish private keys, in spite of $f$ corrupted nodes that may try to conduct man-in-the-middle attacks, an attack that the classical Public key infrastructure (PKI) does not cope with. Once private keys are established encrypted messages can be forwarded over any specific $f + 1$ node independent paths, one of which must be Byzantine free. The Byzantine free path will forward the encrypted message to the receiver while all corrupted messages will be discarded. Since our system should be self-stabilizing, the common private secret should be re-established periodically.

%In addition, a unique random epoch number should be established. The chosen epoch number is used until the next re-establishment period. The epoch number should be encoded as an integral part of each message (encrypted with the message in a non-separable non-constant fashion). Thus, disabling replay attacks by the nodes that are not the sender or the receiver.

\bibliographystyle{abbrv}
\bibliography{bib/RelatedWork}
\onecolumn
\appendix
\FFF
\Section{Correctness of Algorithm~\ref{algo:discovery}}
\label{s:discoveryApp}
%
%----------------------------------------
\noindent {\bf Lemma~\ref{l:BoundedMemory} (Bounded memory)}
{\em Let $p_i \in C$ be a correct node. At any time, there are at most $n \cdot 2^{2 n}$ messages in $InformedTopologyany_i$, where $n = |P|$ and $\BigO(|P| \log (|P|))$ is the message size.
}

\begin{proof}
The queue $InformedTopologyany_i$, is made up of messages in the form $\langle node, neighborhood, visited path \rangle$.
All nodes that appear in the message, i.e., in the first, second or third entry of the tuple are in $N$.
The first entry, i.e. the node name is one of $n$ possibilities. The second and third entries are subsets of $N$. Thus each of them has $2^n$ possibilities. In total there can be at most $2^{n} \cdot 2^{n}  \cdot n$ messages in every $InformedTopologyany_i$.
\end{proof}

%We now turn to the details of the correctness proof (the complete proofs appear in the Appendix and in~\cite{TR}).
%
Definition~\ref{d:legal} specifies the requirements of the network topology discovery task. Definition~\ref{d:CorrectPath} considers correct paths and Definition~\ref{d:IncorruptMessages} considers uncorrupted graph topology messages.

 \begin{definition}[Legal output]
\label{d:legal}
Given correct node $p_i \in C$, we say that $p_i$'s output is {\em legal}, if it encodes graph $G_{output}=(V_{out}, E_{out})$:
(1) $C \subseteq V_{out} \subseteq C \cup B \subseteq N$, and (2) $(E \cap (C \times C)) \subseteq E_{out} \subseteq E \subseteq N \times N$.
 \end{definition}

 \begin{definition}[Correct path]
\label{d:CorrectPath}
We say $path \subseteq N$ is a {\em correct} one if all its nodes are correct, i.e., $path \subseteq C$.
 \end{definition}

 \begin{definition}[Valid message]
\label{d:IncorruptMessages}
In Algorithm~\ref{algo:discovery}, we refer to a message $m=\langle k, Neighborhood_k, VisitedPath_k \rangle$ as a {\em  valid} message when: (1) $p_k \in C$ and $VisitedPath_k$ encodes a correct path in the communication graph, $G$, that starts in $p_k$, and (2) $Neighborhood_k = indices(N_k)$.
 \end{definition}

Lemma~\ref{l:ClearChannel} shows that eventually correct paths do not relay non valid messages. Namely, invalid messages can only exist as the result of: (1) Byzantine
interventions that corrupt messages, or (2) transient faults, which occur only prior to the arbitrary starting configuration considered.~\footnote{This is a common way to argue about self-stabilization, we consider executions that start in an arbitrary configuration that follows the last transient fault, recalling that if additional transient faults occur a new arbitrary configuration is reached from which automatic convergence starts.}

We note that we consider asymptotic behavior in the following lemma and thus, $capacity$ is omitted from the number of asynchronous round until stabilization.
 \begin{lemma}[Eventually valid messages]
\label{l:ClearChannel}
Let $R$ be a fair execution of Algorithm~\ref{algo:discovery} that starts in an arbitrary configuration. Within $\BigO(|N|)$ asynchronous rounds, the system reaches a configuration after which only valid messages are relayed on correct paths.
 \end{lemma}

\begin{proof}
Let $c\in R$ be the starting configuration. Suppose that $c$ includes an invalid message, $m = \langle \ell, Neighborhood_{\ell}, VisitedPath_{\ell} \rangle$, in transit between correct nodes. The lemma is obviously correct for the case that $m$ is relayed by Byzantine nodes during the first $\BigO(|N|)$ asynchronous rounds of $R$. Therefore, we consider only the correct paths, $path$, over which $m$ is relayed during the first $\BigO(|N|)$ asynchronous rounds of $R$. We show that, within $\BigO(|N|)$ asynchronous rounds, no correct node in $path$ relays $m$.

Let $p_j, p_i \in path$ be correct neighbors on the correct path. Suppose that in $c$, message $m$ is in transit from $p_j$ to $p_i$. Upon the arrival of message $m$ to $p_i$ (line~\ref{ln:Receive}), $p_i$ sends $m_i = \langle \ell, Neighborhood_{\ell}, VisitedPath_{\ell} \cup \{ j \} \rangle$ to any neighbor $p_k \in path$ on the path for which $p_k \in N_i \wedge k \not \in VisitedPath_{\ell}$, see line~\ref{c:VisitedPathUpj}.

Node $p_i$ adds $p_j$'s identifier to $m$'s visited path $VisitedPath_\ell$, see line~\ref{c:VisitedPathUpj}. The same argument holds for any correct neighbors, $p^\prime_j, p^\prime_j \in path$ when $p_j$ sends message $m^\prime_j$ to the next node in $path$, node $p^\prime_i$. Therefore, within $|path \setminus VisitedPath_{\ell}|$ asynchronous rounds, it holds that $N^\prime_i \cap (path \setminus VisitedPath_{\ell}) = \{ p^\prime_j, p^\prime_i \}$.

Note that $p^\prime_i$ makes sure that $VisitedPath^\prime_{\ell}$ does not encode loops, i.e., $p_k \not \in VisitedPath^\prime_{\ell}$, see line~\ref{c:VisitedPathUpj}. Therefore, node $p^\prime_i$ does not relay message $m^\prime$ to $p_k$.
\end{proof}

Definition~\ref{d:IncorruptQueue} considers queues that their recent valid messages encode at least $f+1$ vertex-disjoint paths. Moreover, the invalid ones encode at most $f$ such paths.

 \begin{definition}[Valid queue]
\label{d:IncorruptQueue}
Let $p_i, p_k \in C$ be two correct nodes. We say that $p_i$'s queue, $InformedTopology_i$, is {\em valid (with respect to $p_k$)} whenever there is a prefix, $ValidInformation_{i,k}$, of messages $m_k$ in the queue $InformedTopology_i$, such that: (1) there is a subset, $Valid = \{ m_{\ell} = \langle k, Neighborhood_{k}, VisitedPath_{\ell}\rangle : m_{\ell} \textmd{ is valid} \} \subseteq ValidInformation_{i,k}$, for which the set $\{ VisitedPath_{\ell} \}$ encodes at least $f+1$ vertex-disjoint paths, and (2) the set, $Invalid = \{ m_{\ell} = \langle k, Neighborhood_{k}, VisitedPath_{\ell}\rangle : m_{\ell} \textmd{ is invalid}\} \subseteq ValidInformation_{i,k}$, for which the set $\{ VisitedPath_{\ell} \}$ encodes at most $f$ vertex-disjoint paths.
 \end{definition}

Claim~\ref{c:CorrectRelaying} shows that, within $\BigO(|C|)$ asynchronous rounds, correct paths propagate valid messages.

\begin{claim}
\label{c:CorrectRelaying}
Let $path \subseteq C$ be a correct path from $p_i$ to $p_k$. Suppose that $m_i = \langle i, N_{i}, \emptyset \rangle$ is a (valid) message that $p_i$ sends, see line~\ref{ln:sendNeighborhood}. Within $\BigO(|path|)$ asynchronous rounds, message $m_i$ is relayed on $path$, and arrives at $p_k$ as $m_i^\prime = \langle i, N_{i}, path \rangle$. Namely, $path$ is $m_i^\prime$'s visited path.
\end{claim}

%\noindent {\bf Claim~\ref{c:CorrectRelaying}}
%{\em Let $path \subseteq C$ be a correct path from $p_i$ to $p_k$. Suppose that $m_i = \langle i, N_{i}, \emptyset \rangle$ is a (valid) message that $p_i$ sends, see line~\ref{ln:sendNeighborhood}. Within $\BigO(|path|)$ asynchronous rounds, message $m_i$ is relayed on $path$, and arrives at $p_k$ as $m_i^\prime = \langle i, N_{i}, path \rangle$. Namely, $path$ is $m_i^\prime$'s visited path.}

\begin{proofClaim}
Let $c \in R$ be the first configuration that follows the start of $m_i$'s propagation in $path$. I.e., $c$ is the configuration that immediately follows the step in which node $p_i$ sends $m_i$ by executing line~\ref{ln:sendNeighborhood}. Let $p_r, p_j \in path$ be two correct neighbors on the path. Without the loss of generality, suppose that node $p_i$ sends message $m_i$ directly to node $p_r$, i.e., in $c$, node $p_r$ is just about to receive $m_i$. The proof arguments hold also when assuming that $p_j$ sends message $m_j=\langle i, N_{i}, \{ r \} \rangle$ to the next node in $path$. Thus, generality is not lost.

We show that, within one asynchronous round, $p_r$ sends $m_r$ to $p_j$. Upon the arrival of message $m_i$ to $p_r$ (line~\ref{ln:Receive}), node $p_r$ sends the message $m_r$ to any neighbor, such as $p_j$, for which $p_j \in N_r \wedge r \not \in VisitedPath_i =\emptyset$, see line~\ref{c:VisitedPathUpj}. Since the same argument holds when $p_j$ sends $m_j$ to the next node in $path$, we have that within $|path|$ asynchronous rounds, $m_i^\prime$ is delivered to node $p_k$.

%
%Any correct node, when receiving a message from a neighbor, performs line~\ref{c:VisitedPathUpj} correctly. That is, it adds the previous node on the path to the visited path and sends it to all neighbors , one of which is the next node on the path. All this happens within one asynchronous rounds. Thus, after $\BigO(|C|)$, the message is received at the destination with a correct visited path.
\end{proofClaim}

Lemma~\ref{l:IncorruptQueue} shows that queues get to become valid.

 \begin{lemma}[Eventually valid queues]
\label{l:IncorruptQueue}
Let $R$ be a fair execution of Algorithm~\ref{algo:discovery} that starts in an arbitrary configuration and $p_i, p_k \in C$ be any pair of correct nodes. The system reaches a configuration in which the queue, $InformedTopology_i$, is valid (with respect to $p_k$), within $\BigO(|N|)$ asynchronous rounds.
 \end{lemma}

\begin{proof}
Let $c\in R$ be a configuration achieved in Lemma~\ref{l:ClearChannel} within $\BigO(|N|)$ asynchronous rounds. We show that within
$\BigO(|N|)$ asynchronous rounds after $c$, the system reaches a configuration in which $InformedTopology_i$, is {\em valid (with respect to $p_k$)}, see Definition~\ref{d:IncorruptQueue}.

In configuration $c$, all messages in transit on correct paths are valid, see Lemma~\ref{l:ClearChannel}.
Thus, the only messages entering $InformedTopology_i$ are either valid or have passed through Byzantine nodes.  Denote $m_{barrier}$ to be the top message the queue (i.e., the last message entered into the queue) $InformedTopology_i$ in configuration $c$. Moreover, $ValidInformation_{i,k}$ includes all the messages in $InformedTopology_i$, that are between the queue's head and $m_{barrier}$.

We show that condition $(1)$ of Definition~\ref{d:IncorruptQueue} holds. There are $2f+1$ vertex-disjoint paths between $p_i$ and $p_k$. At most $f$ nodes are Byzantine and thus, there are at least $f+1$ vertex-disjoint paths between $p_i$ and $p_k$ that are correct. By Claim~\ref{c:CorrectRelaying} within $\BigO(|C|)$ asynchronous rounds, a valid message, $m_k$, is received on all $f+1$ (correct) vertex-disjoint paths. Message $m_k$ is inserted to $InformedTopology_i$ after configuration $c$. Therefore, $m_k$ is in front of $m_{barrier}$. Hence, the set $Valid = \{ m_{\ell} = \langle \ell, Neighborhood_{\ell}, VisitedPath_{\ell}\rangle : m_{\ell} \textmd{ is valid} \} \subseteq ValidInformation_{i,k}$ contains at least $f+1$ valid messages whose respective visited paths, $VisitedPath_{\ell}$, are vertex-disjoint.

We show that condition $(2)$ of Definition~\ref{d:IncorruptQueue} holds. Any invalid messages, $m_k$, that is sent after configuration $c$, must go through a Byzantine node, see Lemma~\ref{l:ClearChannel}.

\begin{claim}
\label{l:PassThroughByz}
Suppose that message $m$ is relayed through a Byzantine node after configuration $c$, then in any following configuration, while $m$ is still in transit, there is a Byzantine node in the visitedPath.
\end{claim}

\begin{proofClaim}
Observe the first correct node $p_k$ after the last Byzantine node $b$ on $m$'s path. $p_k$ is correct, thus it inserts $b$ to the visited path. $b$ is the last on the path and so the visited path
must contain it until end of transit or passing through a different Byzantine.
\end{proofClaim}

Each such Byzantine node is recorded in the message path, see Claim~\ref{l:PassThroughByz}. Since there are at most $f$ Byzantine nodes, there could be at most $f$ such messages with vertex-disjoint paths. This completes the proof condition $(2)$ and the lemma.
\end{proof}

Lemma~\ref{d:MessageConfirmation} shows that correct information gets confirmed, and requires Definition~\ref{d:MessageConfirmation}.

 \begin{definition}[Message confirmation]
\label{d:MessageConfirmation}
We say that message $m_i = \langle k, Neighborhood_k, VisitedPath_{k_i} \rangle$ is {\em confirmed (by node $p_i$)} when $Neighborhood_k \subseteq ConfirmedTopology_i$.
 \end{definition}

 \begin{lemma}[Eventually confirmed messages]
\label{l:MessageConfirmation}
Let $R$ be a fair execution of Algorithm~\ref{algo:discovery} that starts in an arbitrary configuration and $p_i, p_k \in C$ be any pair of correct nodes. Within $\BigO(|N|)$ asynchronous rounds, the system reaches a configuration after which the fact that message $m_i = \langle k, Neighborhood_k, VisitedPath_{k_i} \rangle$ is confirmed, implies that $Neighborhood_k = indices(N_\ell)$.
 \end{lemma}

\begin{proof}
Let $c \in R$ be the first configuration in which $InformedTopology_i$ is a valid queue and node $p_i$ completes a full iteration of the do forever loop that starts in line~\ref{ln:DoForever}. By Lemma~\ref{l:IncorruptQueue}, the system reaches $c$ within $\BigO(|N|)$ asynchronous rounds.

We know that in configuration $c$, the array $Result_i$ satisfies that $Result_i[k] = indices(N_\ell)$. We go through the computation of $Result$ in lines~\ref{ln:ComputeResults} to~\ref{ln:RemoveGarbage}.

{\bf $\bullet$~~~ $ComputeResults()$,} {\rmfamily line}~\ref{ln:ComputeResults}.~~~
%
%Let $Res_i = ComputeResults()$ (line~\ref{ln:ComputeResults}). We show that $Res_i[k] = indices(N_\ell)$. Assume to the contrary that $Res_i[k] = indices(N^{\prime}_\ell) \not = indices(N_\ell)$.
%
%Since $Res_i[k]$ is $ComputeResults()$'s return value, it must be that the computation within the function gave incorrect results.
%
Let $Res_i[k] = indices(N^{\prime}_\ell)$ be $ComputeResults()$'s return value with respect to node $p_k$. We show that $Res_i[k] = indices(N_\ell)$. Moreover, we show that the neighborhood that will be found will be that which is represented in $Valid = \{ m_{\ell} = \langle k, Neighborhood_{k}, VisitedPath_{\ell}\rangle : m_{\ell} \textmd{ is valid} \} \subseteq ValidInformation_{i,k}$.

We recall that the set $\{ VisitedPath_{\ell} \}$ encodes at least $f+1$ disjoint paths. Also in the prefix $ValidInformation_{i,k}$ one can not find
$f+1$ invalid messages with vertex-disjoint messages; See Definition~\ref{d:IncorruptQueue}.

The function must choose the message containing the neighborhood $Neighborhood_{k}$. Otherwise, we have chosen a different neighborhood for $k$, say $Neighborhood_{k}^{\prime} \not = Neighborhood_{k} = indices(N_k)$. That is, at the time of checking line~\ref{ln:found} with neighborhood
$Neighborhood_{\ell} = Neighborhood_{k}^{\prime}$, there were at least $f+1$ vertex-disjoint paths in $opinion[Neighborhood_{\ell}]$. This is in contradiction to condition $(2)$ of Definition~\ref{d:IncorruptQueue}. Moreover in line~\ref{ln:count}, it holds $Count[k] > f+1$, since at least all the correct paths were counted.

{\bf $\bullet$~~~ $RemoveContradictions()$}, {\rmfamily line}~\ref{ln:RemoveContradictions}.~~~
Let $Res_i = ComputeResults()$ and $ResRemoveContradictions_i = RemoveContradictions(Res_i)$ (line~\ref{ln:RemoveContradictions}). We show that $ResRemoveContradictions_i[r] = indices(N_r)$. The function {\bf $RemoveContradictions()$} modifies $Res_i[r]$ only in line~\ref{ln:resellgetsempty} by nullifying it whenever $Count[r] \leq f$. We demonstrate that, for any correct path $VisitedPath_k$, there exists no $p_\ell$ for which $PathContradictsNeighborhood(p_{\ell}, Res_i[\ell], VisitedPath_k)$ $=$ ${\bf true}$, which is the condition in line~\ref{ln:removeContradictions}.

We explain that there is no node $p_{\ell}$ and a contradicting edge $(p_j, p_{\ell})$ with the set $Res_i[\ell]$.
By the assumption that $VisitedPath_k$ is correct and that node $p_{\ell} \in VisitedPath_k$, we have that $p_{\ell} \in C$ is correct. Thus $Res_i[\ell] = indices(N_\ell)$, see previous item of this claim on $ComputeResults()$. $VisitedPath_k$ is correct, and therefore $( p_j, p_{\ell})$ must be in $VisitedPath_k$.

{\bf $\bullet$~~~ $RemoveGarbage()$,} {\rmfamily line}~\ref{ln:RemoveGarbage}.~~~
This procedure does not modify $Res_i = RemoveContradictions(ComputeResults())$.
We have shown that $Result_i[k] = indices(N_k)$. Thus, only the correct neighborhood is confirmed
for every correct node $p_k$.
\end{proof}

Lemma~\ref{l:NoFake} shows that eventually there are no fake nodes.

 \begin{lemma}[Eventually no fake nodes]
\label{l:NoFake}
Let $R$ be a fair execution of Algorithm~\ref{algo:discovery} that starts in an arbitrary configuration, $p_j \in N$ be any node, and $p_{\ell} \in P \setminus (N)$ be a node that is not included in the communication graph, $G$. Within $\BigO(|N|)$ asynchronous rounds, the system reaches a configuration after which $(p_j, p_\ell) \not \in ConfirmedTopology_i$
 \end{lemma}

\begin{proof}
Let $c \in R$ be the configuration reached within $\BigO(|N|)$ asynchronous rounds according to Lemma~\ref{l:MessageConfirmation}. For any correct node, $p_i \in C$, we show that in $c$, the execution of $RemoveContradictions()$ results in $Count_i[\ell] \leq f$ and nullifies $Result_i[\ell]$.

We start by showing that for every path $p$ that relays a message which encodes the set $Result_i[\ell]$, and does not contain Byzantine nodes, a contradiction is found in $RemoveContradictions()$. Namely, the if conditions of line~\ref{ln:removeContradictions} holds.

Note that, $p$ may not be a correct path even though it contains no Byzantine nodes. For example $p$ may contain nodes $p_z$ that are not even in the communication graph, i.e., $p_z \in P \setminus (N)$.

Let $p_r \in N$ be the first correct node on path $p$. Such a node exists, because $p_i$ is correct and on the path $p$. Since $p_r$ is correct, after the execution of $ComputeResults()$, we have that $p_r$'s neighborhood, $N_r$, is encoded in $Result_i[r]$, see Lemma~\ref{l:MessageConfirmation}.

Denote the last edge in the path $(p_r, p_s)$, where $p_s \in P \setminus (N)$. Note that node $p_s$ is not a node in the system and since $Result_i[r]$ encodes $N_r$'s neighborhood, we have that $p_s  \not \in Result_i[r]$. Thus, the edge $(p_r, p_s)$ is contradicting with the set $Result_i[r]$. Namely, by the condition in line~\ref{ln:removeContradictions}, we have that line~\ref{ln:decreaseCount} must decrease $Count[\ell]$.

%In other words, whenever $Result_i[r]$ encodes an edge that contradicts n a correct path.

We note that immediately before the function $RemoveContradictions()$ returns, the integer $Count[\ell]$ may count only incorrect paths, which contain at least one Byzantine node. Since there are at most $f$ Byzantine nodes, $Count[\ell] \leq f$ as needed.
\end{proof}

Theorem~\ref{l:convenance} demonstrates the self-stabilization properties.

 \begin{theorem}[Self-stabilization]
\label{l:convenance}
Let $R$ be a fair execution of Algorithm~\ref{algo:discovery} that starts in an arbitrary configuration and $p_i \in C$ be a correct node. Within $\BigO(|N|)$ asynchronous rounds, the system reaches a safe configuration after which $p_i$'s output is always legal, see Definition~\ref{d:legal}.
 \end{theorem}

\begin{proof}
The systems reaches configuration $c \in R$ of Lemma~\ref{l:MessageConfirmation} within $\BigO(|N|)$ asynchronous rounds. We show that $c$ is a safe configuration by showing that the output is legal, we must show that $ConfirmedTopology_i$ encodes a graph $G_{output}=(V_{out}, E_{out})$, such that: $(1)$ $C \subseteq V_{out}$, $(2)$ $(E \cap (C \times C)) \subseteq E_{out}$, $(3)$ $V_{out} \subseteq C \cup B \subseteq N$, and $(4)$ $E_{out} \subseteq (E \cap (C \times C)) \cup (B \times (N)) \subseteq P \times N$.

For every correct node $p_k \in C$, we have that $N_k$ is confirmed in $c$, see Lemma~\ref{l:MessageConfirmation}. Thus, $p_k \in V_{out}$ and condition $(1)$ holds.

Let $(p_j, p_k)$ be an edge in the communication graph between two correct nodes, we show $(p_j, p_k) \in E_{out}$. Since $p_j$ is correct, it is inserted to $ConfirmedTopology_i$, see Lemma~\ref{l:MessageConfirmation}. Thus, $(p_j, p_k) \in edges(N_j) \wedge  edges(N_j) \subseteq ConfirmedTopology_i$ in $c$, thus condition $(2)$ holds as well.

There is no $p_{\ell} \in P \setminus (N)$ and node $p_j \in N$, such that $(p_\ell p_j)\in ConfirmedTopology_i$, see Lemma~\ref{l:NoFake}. Thus, $V_{out} \subseteq C \cup B \subseteq N$ and $E_{out} \subseteq (E \cap (C \times C)) \cup (B \times (N)) \subseteq P \times N$. I.e., conditions $(3)$ and $(4)$ hold in $c$.
\end{proof}

\Section{Implementation proposals for $getDisjointPaths()$}
\label{s:ipf}
We consider the problem of relaying messages over the set $CorrectPaths$ when only $ConfirmedTopology$ is known, and propose three implementations to the function $getDisjointPaths()$. The value of $ConfirmedTopology$ is a set of directed edges $(p_i,p_j)$. An undirected edge is approved if both $(p_i,p_j)$ and $(p_j,p_i)$ appear in $ConfirmedTopology$. Other edges in $ConfirmedTopology$ are said to be suspicious. The arguments used here assume that the system is in a safe configuration with respect to Algorithm~\ref{algo:discovery}. For each of the proposed implementations, we show that $|Paths|$ is polynomial and $CorrectPaths \subseteq Paths$. Thus, the sender and the receiver can exchange messages using a polynomial number of paths and message send operations, because each path in $Paths$ is of linear length.

{\bf The case of constant $r$ and $\Delta$.}~~~~ The sender and the receiver exchange messages by using all possible paths between them. This is feasible only when considering $r$-neighborhoods, rather than the entire connected component, where the neighborhood radius, $r$, and the node degree $\Delta$ are constants.

{\bf The case of constant $f$.} ~~~~
This procedure entails sending a message on a path set, $Paths$, where $|Paths|$ is polynomial and $CorrectPaths \subseteq Paths$. %Namely, Moreover, these paths are sufficient to guarantee safe delivery of the message.

%We explain how the sender and the receiver select a set of vertex-disjoint paths, ${\cal P}(p_1,p_2, \ldots p_f) \subseteq Paths$, that contains $f+1$ correct vertex-disjoint paths.

%${\cal P}(p_1,p_2, \ldots p_f) \subseteq Paths$

For each possible choice of $f$ system nodes, $p_1,p_2, \ldots p_f$, the sender and the receiver compute a new graph $G(p_1,p_2, \ldots p_f)$ that is the result of removing $p_1,p_2, \ldots p_f$, from $G_{out}$, which is the graph defined by the discovered topology, $ConfirmedTopology$. Let ${\cal P}(p_1,p_2, \ldots p_f)$ be a set of $f+1$ vertex-disjoint paths in $G(p_1,p_2, \ldots p_f)$ (or the empty set when ${\cal P}(p_1,p_2, \ldots p_f)$ does not exists) and $Paths = \bigcup_{p_1,p_2, \ldots p_f} {\cal P}(p_1,p_2, \ldots p_f)$. We show polynomial message cost by showing that $|Paths|$ is polynomial. We also show that for at least one choice of $p_1,p_2, \ldots p_f$, has a corresponding set ${\cal P}(p_1,p_2, \ldots p_f)$ that contains $CorrectPaths$.

First we show that this procedure only sends messages through a polynomial number of paths. There are $\BigO(n^f)$ possible chooses of $f$ system nodes. Thus, $\BigO(n^f )$ path sets are computed, and since $f$ is a constant, this number is polynomial. Moreover, each such set contains at most $f + 1$ simple paths of linear length, because $p_i$ only computes sets, ${\cal P}(p_1,p_2, \ldots p_f)$, of size $f+1$. Thus, the sender and the receiver can exchange messages using a polynomial number of paths and message send operations.

We show that $CorrectPaths \subseteq Paths$. Consider the permutation choice, $p_1,p_2, \ldots p_f$, in which the set actually contains the set of Byzantine nodes in the system. Thus $G(p_1,p_2, \ldots p_f)$ contains only correct nodes. Furthermore, at least $f+1$ paths that were present in $G_{out}$ are still present in $G(p_1,p_2, \ldots p_f)$, since $G(p_1,p_2, \ldots p_f)$ was obtained from $G_{out}$ by the removal of $f$ (Byzantine) nodes, $p_1,p_2, \ldots p_f$. Hence, there are at least $f+1$ correct vertex-disjoint paths in $G(p_1,p_2, \ldots p_f)$, in ${\cal P}(p_1,p_2, \ldots p_f)$ and in $Paths$.

{\bf The case of no Byzantine neighbors} ~~~~ The procedure assumes that any Byzantine node has no directly connected Byzantine neighbor in the communication graph. Specifically, this polynomial cost solution considers the (extended) graph, $G_{ext}$, that includes all the edges in $confirmedTopology$ and {\em suspicious edges}. Given three nodes, $p_i, p_j, p_k \in P$, we say that node $p_i$ considers the undirected edge $(p_k, p_j)$ suspicious, if the edge appears as a directed edge in $ConfirmedTopology_i$ for only one direction, e.g., $(p_j,p_k)$.

%, see Definition~\ref{d:SuspiciousEdge}.

%
%\begin{definition}[Suspicious edges]
%%
%\label{d:SuspiciousEdge}
%%
%Given three nodes, $p_i, p_j, p_k \in P$, we say that node $p_i$ considers the undirected edge $(p_k, p_j)$ suspicious, if the edge appears as a directed edge in $ConfirmedTopology_i$ for only one direction, e.g., $(p_j,p_k)$.
%\end{definition}
%

The extended graph, $G_{ext}$, may contain fake edges that do not exists in the communication graph, but Byzantine nodes reports on their existence. Nevertheless, $G_{ext}$ includes all the correct paths of the communication graph, $G$. Therefore, the $2f+1$ vertex-disjoint paths that exists in $G$ also exists in $G_{ext}$ and they can facilitate a polynomial cost solution for the message exchange task, as we next show.

Let $G^{\prime}=(N, E_{G^{\prime}})$ be the graph computed from $ConfirmedTopology$ and its suspicious edges. %, see Definition~\ref{d:SuspiciousEdge}.
We demonstrate that $G^{\prime}$'s edges, $E_{G^{\prime}}$, contains the edges, $E_{G}$, of the communication graph, $G$.
Let us consider $e=(p_j, p_k) \in E_{G}$ and show that $e \in E_{G^{\prime}}$.
When both $p_j$ and $p_k$ are correct, the correctness of Algorithm~\ref{algo:discovery} implies $e \in E_{G^{\prime}}$. Suppose that $p_j$ is correct and $p_k$ is Byzantine, and consider the different cases in which $p_k$ decides to report (or not to report) about $e$ as part of its local neighborhood. Namely, either $e \in ConfirmedTopology$, or $e$ is a suspicious edge, because $p_i$ reports about $e$, and $p_k$ decides to report, and respectively, not to report. Since $G \subseteq G^{\prime}$, $G^{\prime}$ must contain $2f+1$ vertex-disjoint paths between any sender $p_s$ and receiver $p_r$, because  $G$ does.
Moreover, the same arguments implies that there may be at most $f$ incorrect paths, which contain each at least one Byzantine node. Hence, there are at least $f+1$ correct vertex-disjoint paths in $Paths$.

\Section{Correctness of Algorithm~\ref{algo:end2end}}
\label{s:end2endproof}

Definitions~\ref{d:MessageConfirmationRec},~\ref{d:MessageApprovedRec} and~\ref{d:ClearSender} are needed for Claim~\ref{l:EventuallyConfirmed}, Claim~\ref{l:EventuallyConfirmedreceiver} and Lemma~\ref{l:RfeAfell}.

\begin{definition}[Confirmation]
\label{d:MessageConfirmationRec}
Given configuration $c$, we say that message $m$ is {\em confirmed (by the receiver)} when
$m \in OutputMessageQueue$.
\end{definition}

\begin{definition}[Approve]
\label{d:MessageApprovedRec}
Given fair execution, $R$, of Algorithm~\ref{algo:end2end}, we say that message $m = \langle Source,$ $Destination,$ $VisitedPath,$ $IntentedPath,$ $ARQLabel,$ $DATA,$ $Payload \rangle$ is {\em being approved (by the sender $p_{Source}$)} during the first atomic step, $a_{sender}$, in which the sender executes line~\ref{ln:ApprovedTrue}, where $Source= sender$ $ARQLabel = m.ARQLabel$ and $Payload = m.Payload$, see line~\ref{ln:DataApprovedTrue}. Denote by $c_{approved}$ the configuration that immediately follows $a_{sender}$. Given configuration $c$ that appears after $c_{approved}$ in $R$, we say that message $m$ {\em is approved (by the sender)} in configuration $c$.
\end{definition}

\begin{definition}[Clear-sender-receiver]
\label{d:ClearSender}
Given configuration $c$, we say that the sender is {\em clear  (with respect to the receiver)}, if the queue $Confirmations[receiver] = \emptyset$  in $c$.
Moreover, the receiver is {\em clear (with respect to the sender)} , if the queue $ReceivedMessages[sender] = \emptyset$ in $c$.
\end{definition}

Claim~\ref{c:correctRelayingEnd2End} shows that a message that is relayed on a correct path is received at the destination within $\BigO(|N|)$ asynchronous rounds. Moreover, the destination receives the message with correct visiting set.
\begin{claim}
Let $R$ be a fair execution of Algorithm~\ref{algo:end2end} that starts in a safe configuration, $c$, with respect to Algorithm~\ref{algo:discovery}. Let $p_{source}, p_{dest} \in C$ be pair of correct nodes. Let $c_{send}$ be the configuration immediately following a step in which $p_{source}$ sends message $Msg$ on a correct path $Path = p_{source}, p_1, p_2, \ldots p_{dest}$  from source, $p_{source}$, to destination, $p_{dest}$. Within $\BigO(|N|)$ asynchronous rounds, $p_{dest}$ receives $Msg$ with a visiting set containing all nodes on $Path$ except $p_{dest}$.
\label{c:correctRelayingEnd2End}
\end{claim}

\begin{proof}
Upon the arrival of message $m$ to $p_k$ (line~\ref{ln:receive1}), node $p_i$ asserts that he is not the destination, $p_{dest}$, (line~\ref{ln:checkDest}). Immediately after, $p_{i}$ sends the message $m$ to the next neighbor, $p_{i+1}$, see line~\ref{ln:notdestsend}. Since the same argument holds when $p_j$ sends $m$ to the next node in $path$, we have that within $|Path|$ asynchronous rounds, $m$ is delivered to node $p_{dest}$.
\end{proof}

Claim~\ref{l:EventuallyConfirmed} says that when the sender repeatedly sends message $Msg$, for a duration of at least $\BigO(|N|)$ asynchronous rounds, the receiver eventually confirms message $Msg$.

\begin{claim}
Let $R$ be a fair execution of Algorithm~\ref{algo:end2end} that starts in a safe configuration, $c$, with respect to Algorithm~\ref{algo:discovery}. Let $p_s, p_r \in C$ be a pair of correct sending and receiving nodes. Suppose that, for a duration of at least $\BigO( capacity \cdot |N|)$ asynchronous rounds, $p_s$'s steps include only the execution of the function $ByzantineFaultTolerantSend(Msg)$ in the loop of line~\ref{ln:FaultSend}. Within that period, the system reaches configuration $c_{receive}$ in which $p_r$ confirms $Msg$.
\label{l:EventuallyConfirmed}
\end{claim}
\begin{proof}
Denote $c_{send}$ as the configuration immediately following the first step in which $p_s$ sends message $Msg$ in $R$, see line~\ref{ln:send}. Within $\BigO( capacity \cdot |N|)$ asynchronous rounds, the first frame containing $Msg$ arrives at $p_r$, see Claim~\ref{c:correctRelayingEnd2End}. Moreover, after another $\BigO( capacity \cdot |N|)$ asynchronous rounds, every correct path relays message $Msg$ at least  $\BigO( capacity \cdot |N|)$  times. This is correct since every asynchronous round, $p_s$ sends a new frame containing $Msg$ on each of the $2f+1$ vertex-disjoint paths. Moreover, by Claim~\ref{c:correctRelayingEnd2End}, the last frame sent on all $2f+1$ paths arrives after another  $\BigO( capacity \cdot |N|)$.

Assume, in the way of proof by contradiction, that $Msg$ is not confirmed by $p_r$. This implies that the queues, $ReceivedMessages[p_s][\ast]$, in $p_r$ containing messages sent from $p_s$ were not cleared at least since $c_{send}$, see line~\ref{ln:clear}. Thus, $p_r$ contains $capacity \cdot n +1$ indications of $Msg$ on $f+1$ vertex-disjoint paths. Denote $c_{last}$ as the configuration immediately after the arrival of the $(capacity \cdot n +1)$-th frame of the $f+1$'th path to relay $capacity \cdot n +1$ frames containing $Msg$. Immediately after $c_{last}$, $p_s$ must go through line~\ref{ln:isEnough}, because the conditions in line~\ref{ln:isEnough} hold. Thus, a contradiction and $Msg$ is confirmed within $\BigO( capacity \cdot |N|)$ asynchronous rounds.
\end{proof}

Claim~\ref{l:EventuallyConfirmedreceiver} says that when the receiver is sending acknowledgments about a message, that message eventually becomes approved.
We note that Claim~\ref{l:EventuallyConfirmedreceiver} considers acknowledgments sent from the receiver to the sender, rather than messages sent from the sender to the receiver, as in Claim~\ref{l:EventuallyConfirmed}. %[@[@Note - check claim references. In the PDF all references in this sentence refer to Claim C]@]@

\begin{claim}
Let $R$ be a fair execution of Algorithm~\ref{algo:end2end} that starts in a safe configuration, $c$, with respect to Algorithm~\ref{algo:discovery}. Let $p_s, p_r \in C$ be a pair of correct sending and receiving nodes. Suppose that, for a duration of at least $\BigO( capacity \cdot |N|)$ asynchronous rounds, $p_r$'s steps include only the execution of the function $ByzantineFaultTolerantSend(Ack)$ in the loop of line~\ref{ln:sendAcks}. That is, $p_r$ is sending acknowledgments on message $Msg$. Within that period, the system reaches configuration $c_{receive}$ in which $p_s$ approves $Msg$, see Definition~\ref{d:MessageApprovedRec}.
\label{l:EventuallyConfirmedreceiver}
\end{claim}
\begin{proof}
Denote $c_{send}$ as the configuration immediately following the first step in which $p_r$ sends acknowledgment $Ack$ in $R$, see line~\ref{ln:sendAcks}. Within $\BigO( capacity \cdot |N|)$ asynchronous rounds, the first frame containing $Ack$ arrives at $p_s$, see Claim~\ref{c:correctRelayingEnd2End}. Moreover, after another $\BigO( capacity \cdot |N|)$ asynchronous rounds, every correct path relays message $Ack$ at least  $\BigO( capacity \cdot |N|)$  times. This is correct since every asynchronous round, $p_r$ sends a new frame containing $Ack$ on each of the $2f+1$ vertex-disjoint paths. Moreover, by Claim~\ref{c:correctRelayingEnd2End}, the last frame sent on all $2f+1$ paths arrives after another  $\BigO( capacity \cdot |N|)$.

The queues, $Confirmations[p_r][\ast]$ are cleared only when a message sent to $p_r$ is approved, see line~\ref{ln:ClearAcks}. Since, $p_r$ is acknowledging the current message, $Msg$, by sending $Ack$, the only message that can be approved is $Msg$. This is true since each path, $Path$, may contain at most $capacity \cdot |N| $ acknowledgments for other messages in the path queues.

Assume, in the way of proof by contradiction, that $Msg$ is not approved by $p_s$. By the arguments above, $p_s$'s queues, $Confirmations_s[p_r][\ast]$, which contains $p_r$'s acknowledgments that $p_s$ received, were not cleared at least since $c_{send}$, see line~\ref{ln:ClearAcks}. Thus, $p_s$ contains $capacity \cdot n +1$ indications of $Ack$ on $f+1$ vertex-disjoint paths. Denote $c_{last}$ as the configuration immediately after the arrival of the $(capacity \cdot n +1)$-th frame of the $f+1$'th path to relay $capacity \cdot n +1$ frames containing $Ack$. Immediately after $c_{last}$, $p_s$ must go through line~\ref{ln:isEnoughConf}, because the conditions in line~\ref{ln:isEnoughConf} hold. Thus, a contradiction and $Msg$ is approved within $\BigO( capacity \cdot |N|)$ asynchronous rounds.
\end{proof}

Lemma~\ref{l:RfeAfell} shows that the senders repeatedly fetch messages.

\begin{lemma}
\label{l:RfeAfell}
Let $R$ be a fair execution of Algorithm~\ref{algo:end2end} that starts in a safe configuration, $c$, with respect to Algorithm~\ref{algo:discovery}. Let $p_s, p_r \in C$ be pair of correct sending and receiving nodes. Moreover, $c_{\ell}$ is the configuration that immediately follows the $\ell$-th time in $R$ in which $p_s$ fetches a message from the input queue. For every $\ell$, the system reaches $c_{\ell}$ within $\BigO(\ell \cdot |N|)$ asynchronous rounds.
\end{lemma}

\begin{proof}
By the code of Algorithm~\ref{algo:end2end}, on every iteration of the do forever loop (lines~\ref{ln:ClearAcks} to~\ref{ln:FaultSend}), a message is fetched in line~\ref{ln:fetch}. This do forever loop includes another loop in line~\ref{ln:FaultSend}. We prove the lemma by showing that the loop of line~\ref{ln:FaultSend} is completed within $\BigO( |N|)$ asynchronous rounds.

The proof considers the case in which the sender, $p_s$, does not wait in line~\ref{ln:FaultSend} for a long time before considering the case in which $p_s$ does wait. We show that for the latter case, the receiver, $p_r$, confirms $p_s$'s current message. After confirming the message, the receiver, $p_r$, begins sending acknowledgments to the sender, $p_s$. The proof shows that after the acknowledgments are sent, $p_s$ approves the message and fetches a new one. We show this by considering the case in which $p_r$ repeatedly sends acknowledgments for a sufficient amount of time, and a case in which it does not.

Suppose that $p_s$ does not wait in line~\ref{ln:FaultSend} more than $\BigO(capacity \cdot |N|)$ asynchronous rounds. In this case, $p_s$ starts the infinite loop again within $\BigO(capacity \cdot |N|)$ asynchronous rounds, and fetch a new message, see line~\ref{ln:fetch}. Thus, for the case in which $p_s$ does not wait in line~\ref{ln:FaultSend} more than $\BigO(capacity \cdot |N|)$ asynchronous rounds, the lemma is correct.

Suppose that $p_s$ is executing line~\ref{ln:FaultSend} and waits for acknowledgments on message $Msg$ for more than $\BigO(capacity \cdot |N|)$ asynchronous rounds. Thus, $p_s$ floods $2f+1$ vertex-disjoint paths with the message $Msg$, see Figure~\ref{fig:procedures}. Eventually, the receiver, $p_r$, receives message $Msg$ for $\BigO(capacity \cdot |N|)$ times on $f+1$ vertex-disjoint paths and confirms $Msg$, see Claim~\ref{l:EventuallyConfirmed}. After confirming it, the receiver sends acknowledgments on $2f+1$ vertex-disjoint paths until confirming a new message $Msg_{new}$. This is true because the condition in line~\ref{ln:sendAcks} holds only when a new message is confirmed, see line~\ref{ln:newMessageTrue}.

Let us consider the case in which, during $\BigO(capacity \cdot |N|)$ asynchronous rounds, message $Msg_{new}$ does not arrive to the receiver. By Claim~\ref{l:EventuallyConfirmedreceiver}, eventually the sender receives the acknowledgments for $capacity \cdot n + 1$ times on $f+1$ vertex-disjoint paths. Claim~\ref{l:EventuallyConfirmedreceiver} also says that the sender considers the message accepted by the receiver. In line~\ref{ln:ApprovedTrue}, the sender assigns $Approved = {\bf true}$. Thus, the condition in line~\ref{ln:FaultSend} holds and the sender fetches the next message, see line~\ref{ln:fetch}. Hence, the system reaches configuration $c_{fetch}$ that immediately follows a step in which the sender, $p_s$, fetches the next message. Thus, for the case in which, during $\BigO(capacity \cdot |N|)$ asynchronous rounds, message $Msg_{new}$ does not arrive to the receiver, the lemma is correct.

We continue by considering the case in which, during $\BigO(capacity \cdot |N|)$ asynchronous rounds, message $Msg_{new}$ does arrive to the receiver. Let $c_{conf}$ be the configuration that immediately follows the step in which $p_r$ confirms $Msg$. Since the receiver confirms $Msg$, we have that $p_r$ is clear (with respect to the sender) in configuration $c_{conf}$, see Definition ~\ref{d:ClearSender} and line~\ref{ln:clear}.

If $Msg_{new}$ was sent by the sender, it must have been fetched after $c$, and $c_{fetch}$ is reached when message $Msg_{new}$ is fetched.
It may be the case however, that $Msg_{new}$ was not sent by the sender.  Message $Msg_{new}$ was confirmed by $2f+1$ vertex-disjoint paths. Since there are at most $f$ Byzantine nodes, at least one of these paths, $Path$, must be correct. Moreover, in $c_{conf}$, the receiver is clear, thus the $capacity \cdot n + 1$ that $p_r$ counts in $ReceivedMessages[p_s][\ast]$ have all been received after configuration $c_{conf}$.
Note that the sender sends at least one of these messages, because at most $capacity \cdot n $ messages could be in the edges of $Path$ at any given configuration. Thus the sender sends $Msg_{new}$, which $p_s$ fetches immediately before $c_{fetch}$. I.e., the system reaches $c_{fetch}$.
\end{proof}

Theorem~\ref{l:convenance} says that, starting from the fourth (or even the third) message that the sender fetches, the receiver confirms the sender's messages. The proof of Theorem~\ref{l:convenance} is based on Lemma~\ref{l:ClearChannelend2endB}, which says that, in every sequence of four messages that the sender is fetching, the receiver confirms the fourth (or even the third) message.

\begin{lemma}
\label{l:ClearChannelend2endB}
Let $R$ be a fair execution of Algorithm~\ref{algo:end2end} that starts in a safe configuration, $c_{start}$, with respect to Algorithm~\ref{algo:discovery}. Let $c_{h}$ be a configuration that immediately follows the $h$-th step in which the sender fetches the $h$-th input queue message, $m_h$. Within $\BigO(|N|)$ asynchronous rounds, the receiver confirms message $m_4$.
\end{lemma}

%\noindent {\bf Lemma~\ref{l:ClearChannelend2endB}}
%{\em
%}

\begin{proof}

\begin{claim}
\label{l:Bcstc}
In $c_2$, the sender is clear (with respect to the receiver), see Definition~\ref{d:ClearSender}.
\end{claim}

\begin{proofClaim}
By definition, $c_2$ immediately follows atomic step $a_2$, in which, after clearing the confirmation queue in line~\ref{ln:ClearAcks}, the sender fetches message $m_2$ and sends it.
\end{proofClaim}

\begin{claim}
\label{l:sthdlofp}
Between the configurations $c_{3}$ and $c_4$, there is a configuration $c_{receiver-clear}$ in which the receiver is clear (with respect to the sender).
\end{claim}

\begin{proofClaim}
Suppose, without the loss of generality, that immediately after $c_{sender-clear}$, the sender is waiting for a message with label $1$. By lemma~\ref{l:RfeAfell}, the sender eventually fetches the next message. The sender can only fetch a new message once $Approved$ is true, see line~\ref{ln:FaultSend}. Moreover, $Approved$ is only set to $true$ once the queue $Confirmations[receiver][\ast]$ contains $2f+1$ flooded paths, see line~\ref{ln:isEnoughConf}. Thus, the sender counts $2f+1$ vertex-disjoint paths that relayed acknowledgments with label $1$. Moreover, the sender is clear in $c_{sender-clear}$. Hence, configuration $c_{sender-clear}$ contains no message in $Confirmations[receiver][\ast]$. Starting from $c_{sender-clear}$, the sender receives $capacity \cdot n + 1$ acknowledgments on $2f+1$ vertex-disjoint paths for the current message with label $1$.
Note that at least one of these $2f+1$ paths, $Path$, is correct, because there are $f$ Byzantine. Since $|Path| \leq n$ and each edge on $Path$ may contain at most $capacity$ messages, we have that at least one of the acknowledgments that includes $Path$ as its visiting path, is sent by the receiver between $c_{sender-clear}$ and configuration $c_{receiver-send} \in R$. We show that $c_{receiver-send} = c_{receiver-clear}$.

This means that after $c_{sender-clear}$, the sender clears the confirmations queue, $Confirmations[receiver][\ast]$, and fetches the next message, assigning it the label $2$, see lines~\ref{ln:ClearAcks} through line~\ref{ln:FaultSend}. By similar arguments, we know that the receiver sends at least one acknowledgment with label $2$.

To conclude, there is a configuration $c \in R$ in which the receiver is sending acknowledgments with label $1$, and then a configuration $c^{\prime}$ in which the receiver sends acknowledgments with label $2$. Moreover, between two consecutive executions of line~\ref{ln:sendAcks}, the receiver has to go through line~\ref{ln:clear}. Thus, the receiver cleared it's message queues, $Confirmations[sender][\ast]$, immediately before configuration $c_{receiver-clear}$ and $c_{receiver-send} = c_{receiver-clear}$.
\end{proofClaim}

Let us consider configuration $c_{receiver-clear}$ from the end of proof of Claim~\ref{l:sthdlofp}.

The next message to be sent after $c_{receiver-clear}$, is $m_4$, the message fetched in $c_4$, with label $0$. Between $c_{receiver-clear}$ and $c_4$, all messages sent by the sender have the label $2$. By arguments stated above, the message, $m$, that is the next message to be confirmed after $c_{receiver-clear}$, must have been sent by the sender at least once since $c_{receiver-clear}$. The sender, sends only messages with label $0$ and $2$. Moreover, the last message to be confirmed had a label $2$. Thus, $CurrentLabel = 2$, see line~\ref{ln:AssignLabel}. Any sent message with label $2$ is not inserted to the confirmations queue, $Confirmations[sender][\ast]$ between $c_{receiver-clear}$ and the configuration that immediately follows the next sender's fetch, see line~\ref{ln:deleiverMessage}. Thus, by line~\ref{ln:changeLabel}, the next message to be confirmed is a message with label $0$, which must be $m_4$.
\end{proof}

\noindent {\bf Theorem~\ref{l:convenance} (Self-stabilization)}
{\em Let $R$ be a fair execution of Algorithm~\ref{algo:end2end} that starts in an arbitrary configuration. Within $\BigO(|N|)$ asynchronous rounds, the system reaches a safe configuration $c$ after which:
(1) for every step $a^{m}_{s}$ where the sender sends $m$ there is a corresponding step $a^m_{r} \in R$ where the receiver confirms message $m$, and
(2) for every step $a^m_{r}$, there is a corresponding step, $a^m_{s} \in R$, that occurs before $a^m_{r}$ and in which the sender sends $m$.
}

\begin{proof}
Let $c$ be the configuration that Claim~\ref{l:sthdlofp} denote as $c_4$, which the system reaches within $\BigO(|N|)$ asynchronous rounds, see Lemma~\ref{l:RfeAfell}. Let $m_i$ be the $i$-th message fetched.

Suppose that $i \geq 4$.  Lemma~\ref{l:ClearChannelend2endB} considers the four consecutive messages $m_{i-3}, \ldots m_i$ and says that the receiver confirms message $m_i$. Thus, condition (1) holds.

Condition (2) follows from arguments similar to the ones used in the proof of Claim~\ref{l:EventuallyConfirmed}. Namely, for the case of $i \geq 5$, message $m_{i-1}$ is confirmed, see lemma~\ref{l:ClearChannelend2endB}. Immediately after the receiver confirms $m_{i-1}$, it clears the queue $ReceivedMessages[sender][\ast]$, see lines~\ref{ln:deleiverMessage} to~\ref{ln:clear}. Thus, there exists a configuration $c_{receiver-clear}$ in which the receiver is clear (with respect to the sender) before $c_i$, see Definition~\ref{d:ClearSender}. Moreover, a message is confirmed only if the queue $ReceivedMessages[sender][\ast]$ contains $2f+1$ flooded paths, see line~\ref{ln:isEnough}. These flooded paths implies that in configuration $c_i$, the queue $ReceivedMessages[sender][\ast]$ contains $capacity \cdot n + 1$ indications of $m_i$ on $2f+1$ node disjoint paths. Thus, $m_i$ is confirmed only after a period that follows $c_{receiver-clear}$ and includes its reception at least $capacity \cdot n +1 $ times on each of the $2f+1$ vertex-disjoint paths.

Recall that we assume that there are at most $f$ Byzantine nodes in the system. At least one path, $Path$, of the above $2f+1$ paths is correct.  Moreover, $|Path| \leq n$ and each edge on $Path$ may contain at most $capacity$ messages. Thus, at least one of the $capacity \cdot n +1 $ message that were relayed on the correct path $Path$ was sent by the sender. This completes the correctness proof.
\end{proof}

\end{document}